\documentclass[11pt]{article}
\pdfoutput=1

\usepackage[T1]{fontenc}
\usepackage[utf8]{inputenc}
\usepackage{fullpage}
\usepackage{cite}

\usepackage[activate={true,nocompatibility},final,tracking=true,kerning=true,spacing=true]{microtype}
\microtypecontext{spacing=nonfrench}

\usepackage{amssymb}
\usepackage{amsmath}

\usepackage{xcolor}
\usepackage{hyperref}

\usepackage{tikz}
\usetikzlibrary{shapes.geometric,shapes.arrows,patterns}
\usepackage{pgfplots}
\pgfplotsset{compat=newest}
\usepackage{subfigure}

\usepackage{booktabs}
\usepackage{multirow}

\usepackage{xspace}

\newcommand{\seq}[1]{\left\langle #1\right\rangle}
\newcommand{\seqGilt}[2]{\left\langle #1\gilt #2\right\rangle}
\newcommand{\Id}[1]{\ensuremath{\text{{\sf #1}}}}
\newcommand{\ceil}[1]{\left\lceil #1\right\rceil}
\newcommand{\set}[1]{\left\{ #1\right\}}
\newcommand{\gilt}{:}
\newcommand{\card}[1]{\left\vert{#1}\right\vert}
\newcommand{\abs}[1]{\left| #1\right|}
\newcommand{\nat}{\mathbb{N}}
\newcommand{\punkt}{\enspace .}
\newcommand{\Def}{:=}

\newcommand{\prob}[1]{{\mathbf{P}}\left[#1\right]}

\newcommand{\expect}[1]{{\mathbf{E}}\left[#1\right]}

\newcommand{\Oh}[1]{\mathcal{O}\!\left( #1\right)}
\newcommand{\Ohsmall}[1]{\mathcal{O}(#1)}
\newcommand{\Th}[1]{\Theta\!\left( #1\right)}
\newcommand{\Thsmall}[1]{\Theta(#1)}
\newcommand{\Om}[1]{\Omega\left(#1\right)}
\newcommand{\Omsmall}[1]{\Omega(#1)}

\newenvironment{code}{\noindent\normalsize\begin{tabbing}\hspace{2em}\=\hspace{2em}\=\hspace{2em}\=\hspace{2em}\=\hspace{2em}\=\hspace{2em}\=\hspace{2em}\=\hspace{2em}\=\hspace{2em}\=\hspace{2em}\=\kill}{\end{tabbing}}
\newcommand{\Of}{\ensuremath{\text{\bf of\ }}}
\newcommand{\Declare}[2]{#1\mbox{ \rm : }#2}
\newcommand{\Function} {{\bf Function\ }}
\newcommand{\Funct}[3]{\Function #1\Declare{{\rm (}{#2\rm )}}{#3}}
\newcommand{\Procedure}{{\bf Procedure\ }}
\newcommand{\Is}{\mbox{\rm := }}
\newcommand{\If}       {{\bf if\ }}
\newcommand{\Then}     {{\bf then\ }}
\newcommand{\Else}     {{\bf else\ }}
\newcommand{\Return}   {{\bf return\ }}
\newcommand{\Rem}[1]   {{\bf (*~}{\rm#1}{\bf ~*)}}
\newcommand{\RRem}[1]   {\`{\bf --\hspace{0.5mm}--~}{\rm#1}}
\newcommand{\Do}       {{\bf do\ }}
\newcommand{\For}      {{\bf for\ }}
\newcommand{\To}       {{\bf to\ }}
\newcommand{\Elsif}    {{\bf elsif\ }}
\newcommand{\Until}    {{\bf until\ }}

\newcommand{\sq}{\hbox{\rlap{$\sqcap$}$\sqcup$}}
\newcommand{\qed}{\hspace*{\fill}\sq}
\newenvironment{proof}{\noindent {\bf Proof.}\ }{\qed\par\vskip 4mm\par}
\newdimen\endofsize\endofsize=0.5em

\newtheorem{theorem}{Theorem}
\newtheorem{lemma}[theorem]{Lemma}
\newtheorem{corollary}[theorem]{Corollary}

\newcommand{\Tstart}{\alpha}
\newcommand{\Tword}{\beta}
\newcommand{\deleteMinPar}{\Id{deleteMin}\ensuremath{^*}}
\def\eps{\varepsilon} 
\def\epsilon{\varepsilon} 
\newcommand{\Error}{\tilde\eps}
\newcommand{\sfrac}[2]{\smash{\frac{#1}{#2}}}

\newcommand{\fakepar}[1]{\bigskip\par\noindent\textbf{#1}}

\hypersetup{
  pdftitle={Communication Efficient Algorithms for Top-k Selection Problems},
  pdfauthor={Lorenz Hübschle-Schneider, Peter Sanders, and Ingo Müller},
  pdfsubject={branch-and-bound, data redistribution, frequent elements, priority queue, sampling, selection, threshold algorithm},
  colorlinks=true,
  pdfborder={0 0 0},
  bookmarksopen=true,
  bookmarksopenlevel=1,
  bookmarksnumbered=true,
  linkcolor=blue!60!black,
  citecolor=blue!60!black,
  urlcolor=blue!60!black,
  filecolor=green!60!black,
  pdfpagemode=UseNone,
  unicode=true,
}

\newcommand{\frage}[2][?]{\impmark{#1}{[\textcolor{red}{\textbf{#2}}]}}

\usepackage{xparse}
\NewDocumentCommand\todo{G{*}}{\frage[TODO]{#1}}

\usepackage[normalem]{ulem}
\newcommand{\impmark}[1]{\strut\vadjust{\domark{#1}}}
\newcommand{\domark}[1]{\vbox to 0pt{\kern-\dp\strutbox\smash{\llap{\textbf{\textcolor{red}{#1}}\kern1em}}\vss}}
\makeatletter
\font\uwavefont=lasyb10 scaled 652
\def\dodgyText{\bgroup \markoverwith{\lower3.5\p@\hbox{\uwavefont \textcolor{red}{\char58}}}\ULon}
\makeatother
\usepackage{xifthen}
\newcommand{\dodgy}[2][]{\impmark{\ifthenelse{\isempty{#1}}{***}{* #1}}\dodgyText{#2}}

\begin{document}

\title{Communication Efficient Algorithms for\\ Top-k Selection Problems}

\author{Lorenz Hübschle-Schneider$^*$ \\ \href{mailto:huebschle@kit.edu}{huebschle@kit.edu} \and Peter Sanders$^*$ \\ \href{mailto:sanders@kit.edu}{sanders@kit.edu} \and Ingo Müller$^{*\dagger}$ \\ \href{mailto:ingo.mueller@kit.edu}{ingo.mueller@kit.edu}\and\\
$^*$ Karlsruhe Institute of Technology, Karlsruhe, Germany\\
$^\dagger$ SAP SE, Walldorf, Germany}

\date{}

\maketitle


\begin{abstract}
We present scalable parallel algorithms with sublinear per-processor
communication volume and low latency for several fundamental problems related
to finding the most relevant elements in a set, for various notions of
relevance: We begin with the classical selection problem with unsorted input.
We present generalizations with locally sorted inputs, dynamic content (bulk-parallel priority queues),
and multiple criteria. Then we move on to finding frequent objects and top-$k$ sum aggregation.
Since it is unavoidable that the output of these algorithms might
be unevenly distributed over the processors, we also explain how to
redistribute this data with minimal communication.

\vspace*{4mm}\noindent
\end{abstract}



\section{Introduction}

Overheads due to communication
latency and bandwidth limitations of the communication networks are one of the
main limiting factors for distributed computing.  Parallel algorithm theory has
considered the latency part of this issue since its beginning.  In particular,
execution times polylogarithmic in the number~$p$ of processing elements (PEs)
were a main focus.  Borkar~\cite{Borkar13} argues that an exascale computer
could only be cost-effective if the communication capabilities (bisection width)
scale highly sublinearly with the amount of computation done in the connected
subsystems.  Google confirms that at data center scale, the network is the most
scarce resource and state that they ``don't know how to build big networks that
deliver lots of bandwidth''~\cite{Vahdat15}.  In a previous paper~\cite{SSM13}
we therefore proposed to look more intensively for algorithms that require
bottleneck communication volume sublinear in the local input size.  More
precisely, consider an input consisting of $n$ machine words distributed over
the PEs such that each PE holds $\Oh{n/p}$ words.  Then sublinear communication
volume means that no PE sends or receives more than $o(n/p)$ machine words of
data.

Here, we combine the latency and data volume aspects.  We consider some
fundamental algorithmic problems that have a large input (size~$n$) and a
relatively small output (size~$k$).  Since potentially many similar problems
have to be solved, both latency and communication volume have to be very
small---ideally polylogarithmic in~$p$ and the input parameters.  More
precisely, we consider problems that ask for the $k$ most ``relevant'' results
from a large set of possibilities, and aim to obtain low bottleneck
communication overhead.

In the simplest case, these are totally ordered elements and we ask for
the~$k$ smallest of them---the classical selection problem. Several variants of
this problem are studied in Section~\ref{s:selection}.  For the classical
variant with unsorted inputs, a careful analysis of a known algorithm~\cite{San98a-e}
shows that the previously assumed random allocation of the
inputs is actually not necessary and we get running time $\Ohsmall{\sfrac{n}{p}+\log p}$.
For locally sorted input we get latency~$\Ohsmall{\log^2kp}$. Interestingly,
we can return to logarithmic latency if we are willing to relax the output size~$k$
by a constant factor. This uses a new technique for obtaining a pivot element
with a given rank that is much simpler than the previously proposed techniques
based on sorting.

A data structure generalization of the selection problem are bulk-parallel priority queues.
Previous parallel priority queues are not communication efficient in the sense
that they move the elements around, either by sorting~\cite{DeoPra92} or by
random allocation~\cite{San98a-e}.  Section~\ref{s:pq} generalizes the results on selection from
Section~\ref{s:selection}. The key to making this work is to use an appropriately
augmented search tree data structure to efficiently support insertion, deletion,
and all the operations needed by the parallel selection algorithm.

A prominent problem in information retrieval is to extract the $k$ most
relevant objects (e.g.\ documents), where relevance is determined by a
monotonous function that maps several individual scoring functions to the
overall relevance.  For each individual score, a list of objects is precomputed
that stores the objects in order of decreasing score.  This is a multicriteria
generalization of the sorted top-$k$ selection problem discussed in
Section~\ref{s:selection}.  In Section~\ref{s:multic} we parallelize
established sequential algorithms for this problem.
The single-criterion selection algorithms are used as subroutines---%
the sorted one for approximating the list elements scanned by the sequential
algorithm and the unsorted one to actually identify the output. The algorithm
has polylogarithmic overhead for coordinating the~PEs and manages to contain unavoidable
load imbalance in a single phase of local computation. This is achieved with a
fast estimator of the output size generated with a given number of scanned
objects.

A fundamental problem in data mining is finding the most frequently occurring
objects.  This is challenging in a distributed setting since the globally most
frequent elements do not have to be locally frequent on any particular~PE.  In
Section~\ref{s:topk} we develop very fast sampling-based algorithms that find a
$(\eps,\delta)$-approximation or \emph{probably approximately correct answer},
i.e., with probability at least~$1-\delta$ the output is correct within~$\eps n$.
The algorithms run in time logarithmic in~$n$,~$p$, and~$1/\delta$.  From a
simple algorithm with running time factor~$1/\eps^2$ we go to more sophisticated
ones with factor~$1/\eps$.  We also show how to compute the exact result
with probability at least~$1-\delta$ if the elements are non-uniformly
distributed.

Subsequently, we generalize these results to sum aggregation, where object
occurrences are associated with a value.  In Section~\ref{s:sumagg}, we are thus
looking for the objects whose values add up to the highest sums.

All of the above algorithms have the unavoidable limitation that the output may be
unevenly distributed over the~PEs for general distributions of the input
objects.  This can lead to load imbalance affecting the efficiency of the
overall application. Offsetting this imbalance will require communication so
that one might argue that striving for communication efficiency in the selection
process is in vain.  However, our methods have several advantages over non-communication
efficient selection algorithms. First of all, priorities can be ignored during
redistribution since all selected elements are relevant. Hence, we can employ
any data redistribution algorithm we want. In particular, we can use an adaptive
algorithm that only moves data if that is really necessary. In
Section~\ref{s:redistribution}, we give one such algorithm that combines prefix
sums and merging to minimize data movement and incurs only logarithmic
additional delay. Delegating data movement to the responsibility of the
application also has the advantage that we can exploit properties of the
application that a general top-$k$ selection algorithm cannot. For
example, multicriteria selection algorithms as discussed in
Section~\ref{s:multic} are used in search engines processing a stream of many
queries. Therefore, it is enough to do load balancing over a large number of
queries and we can resolve persistent hot spots by adaptive data migration or
replication. Another example are branch-and-bound
applications~\cite{KarZha93,San98a-e}.  By randomizing necessary data
migrations, we can try to steer the system away from situations
with bad data distribution
(a proper analysis of such an algorithm is an open problem though).

Some of our main results are listed in Table~\ref{tbl:results}.
Refer to Appendix~\ref{app:prevproof} for proofs and further discussion.

\newcommand{\Problem}[1]{\multirow{2}{3.4cm}{#1}}
\newcommand{\ML}[1]{\multirow{2}{*}{#1}}
\newcommand{\Row}[3]{\Problem{#1} & \ML{~~#2~~} & \ML{~~#3~~} \\\\[0.5em]}
\newcommand{\LastRow}[3]{\Problem{#1} & \ML{~~#2~~} & \ML{~~#3~~} \\\\}

\begin{table}[bt]
  \caption[]{Our main results.  Parameters: input size $n$;
    output size~$k$; number of PEs~$p$; startup overhead~$\Tstart$;
    communication cost per word~$\Tword$; relative error~$\eps$;
    failure probability~$\delta$.  Monitoring queries are marked with $^\dag$.}
  \label{tbl:results}
\centering
\begin{tabular}{l cc}
\toprule
\textbf{Problem} & \multicolumn{2}{c}{\textbf{Asymptotic running time in our model} $\Oh{\cdot}$}\\&old&new\\\midrule
\Row{Unsorted Selection}
{$\Om{\Tword\frac{n}{p}} + \Tstart \log{p}$ \cite{San98a-e}}
{$\frac{n}{p} + \Tword \min\!\left(\!\sqrt{p}\frac{\log{n}}{\log{p}}, \frac{n}{p}\!\right) + \Tstart \log{p}$}

\Row{Sorted Selection}
{$\Om{k \log{n}}$ (seq.) \cite{VSIR91}}
{$\Tstart \log^2{\!kp}$ | $\Tstart \log{kp}$ (flexible $k$)}

\Row{Bulk Priority Queue \textsf{insert*} + \textsf{deleteMin*}}
{$\log{\frac{n}{k} + \Tstart\!\left(\!\frac{k}{p} + \log{p}\!\right)}$ \cite{San98a-e}}
{$\Tstart \log^2{\!kp}$ | $\Tstart \log{kp}$ (flexible $k$)}

\Row{Heavy Hitters}
{$\frac{n}{p} + \Tstart\frac{\sqrt{p}}{\eps}\log{n}$ \cite{HuaYiZha2012} $^\dag$ for~$\delta=\,$const.}
{cf. top-$k$ most frequent objects}

\Row{Top-$k$ Most Frequent Objects}{$\Om{\frac{n}{p} + \Tword \frac{k}{\eps} + \Tstart\frac{1}{\eps}}$ \cite{BabOls03short} $^\dag$}
{$\frac{n}{p} + \Tword\frac{1}{\eps} \sqrt{\frac{\log{p}}{p} \log{\frac{n}{\delta}}} + \Tstart\log{n}$}

\Row{Top-$k$ Sum Aggregation}
{$\frac{n}{p}+\Tword\frac{1}{\eps}\sqrt{p \log\frac{n}{\delta}} + \Tstart p$ (centralized)}
{$\frac{n}{p}+\Tword\frac{\log{p}}{\eps}\sqrt{\frac{1}{p}\log\frac{n}{\delta}} + \Tstart\log{n}$}

\LastRow{Multicriteria Top-$k$}{not comparable}{$\log{K}(m^2\log K +\Tword m +\Tstart\log p)$}

\bottomrule
\end{tabular}
\end{table}

\section{Preliminaries}\label{s:prelim}

Our input usually consists of a multiset~$M$ of~$\card{M}=n$ objects, each
represented by a single machine word.  If these objects are ordered, we assume
that their total ordering is unique.  This is without loss of generality since
we can make the value~$v$ of object~$x$ unique by replacing it by the pair
$(v,x)$ for tie breaking.

Consider~$p$ processing elements (PEs) connected by a network.  PEs are numbered
$1..p$, where~$a..b$ is a shorthand for~$\set{a,\ldots,b}$ throughout this
paper.  Assuming that a PE can send and receive at most one message at a time
(full-duplex, single-ported communication), sending a message of size~$m$
machine words takes time~$\Tstart+m\Tword$.  We often treat the time to initiate
a connection~$\Tstart$ and to send a single machine word~$\Tword$ as variables
in asymptotic analysis.  A running time of~$\Oh{x+\Tword y+\Tstart z}$ then
allows us to discuss \emph{internal work}~$x$, \emph{communication volume}~$y$,
and \emph{latency}~$z$ separately.  Depending on the situation, all three
aspects may be important, and combining them into a single expression for
running time allows us to specify them concisely.

We present some of the algorithms using high level pseudocode in a
\emph{single program multiple data} (SPMD) style---the same algorithm runs on
each~PE, which perform work on their local data and communicate predominantly through
collective operations. Variables are local by default. We can
denote data on remote~PEs using the notation~$x@i$, which refers to the value of
variable~$x$ on~PE~$i$. For example,~$\sum_i x@i$ denotes the global sum of~$x$
over all~PEs, which can be computed using a sum (all-)reduction.

\paragraph*{Collective communication.}
\emph{Broadcasting} sends a message to all~PEs.  \emph{Reduction} applies an
associative operation (e.g., sum, maximum, or minimum) to a vector of
length~$m$.  An \emph{all-reduction} also broadcasts this result to all~PEs.  A
\emph{prefix-sum} (scan) computes~$\sum_{i=1}^{j}x@i$ on~PE~$j$ where~$x$ is a
vector of length~$m$.  The \emph{scatter} operation distributes a message of
size~$m$ to a set of~$y$~PEs such that each of them gets a piece of size~$m/y$.
Symmetrically, a \emph{gather} operations collects~$y$ pieces of size~$m/y$ on a
single~PE.\@ All of these operations can be performed in
time~$\Oh{\Tword m + \Tstart\log p}$~\cite{BalEtAl95,SST09}.  In an all-to-all
personalized communication (all-to-all for short) each~PE sends one message of
size~$m$ to every other~PE.\@ This can be done in time
$\Oh{\Tword mp+\Tstart p}$ using direct point-to-point delivery or in time
$\Oh{\Tword mp\log p+\Tstart \log p}$ using indirect delivery \cite[Theorem
3.24]{Lei92}.
The all-to-all broadcast (aka gossiping) starts with a
single message of size~$m$ on each~PE and ends with all~PEs having all these
messages. This operation works in time~$\Oh{\Tword mp+\Tstart\log p}$.

\paragraph*{Fast inefficient Sorting.}
Sorting $\Oh{\sqrt{p}}$ objects can be done in time $\Oh{\Tstart\log p}$ using
a brute force algorithm performing all pairwise object comparisons in parallel
(e.g.,~\cite{ABSS15}).

\paragraph*{Search trees.}
Search trees can represent a sorted sequence of objects such that a variety of
operations can be supported in logarithmic time.  In particular, one can
insert and remove objects and search for the next largest object given a key.
The operation~$T.\Id{split}(x)$ splits T into two search trees~$T_1,T_2$ such
that~$T_1$ contains the objects of~$T$ with key~$\leq x$ and~$T_2$ contains the
objects of~$T$ with key~$> x$. Similarly, if the keys in a tree~$T_1$ are
smaller than the keys in~$T_2$ then~$\Id{concat}(T_1,T_2)$ returns a search tree
representing the concatenation of the two sequences.  If one additionally stores
the size of subtrees, one can also support the select operation~$T[i]$, which
returns the~$i$-th largest object in~$T$, and operation~$T.\Id{rank}(x)$,
which returns the number of objects in~$T$ of value at most~$x$. For an example
of such a data structure see~\cite[Chapter~7]{MehSan08}.

\paragraph*{Chernoff bounds.} We use the following Chernoff
bounds~\cite{MehSan08,MU05} to bound the probability that a sum~$X=X_1+\ldots+X_n$ of~$n$
independent indicator random variables deviates substantially from its expected value~$\expect{X}$. For~$0<\varphi<1$, we have
\begin{align}
&\prob{X < (1-\varphi)\expect{X}} \leq e^{-\frac{\varphi^2}{2}\expect{X}}\label{eq:ch_leq} \\
&\prob{X > (1+\varphi)\expect{X}} \leq e^{-\frac{\varphi^2}{3}\expect{X}}\label{eq:ch_geq} \punkt
\end{align}

\paragraph*{Bernoulli sampling.}
A simple way to obtain a sample of a set~$M$ of objects is to select each object
with probability~$\rho$ independent of the other objects.  For small sampling probability~$\rho$, the
naive implementation can be significantly accelerated from time~$\Oh{\card{M}}$ to
expected time $\Oh{\rho\card{M}}$ by using skip values---a value of~$x$ indicates
that the following~$x-1$ elements are skipped and the~$x$-th element
is sampled. These skip values follow a geometric
distribution with parameter~$\rho$ and can be generated in constant time.

\paragraph*{Distributed FR-Selection \rm\normalsize\cite{San98a-e}}\label{s:twoPivot}
Floyd and Rivest~\cite{FloRiv75} developed a modification of quickselect using
two pivots in order to achieve a number of comparisons that is close to the
lower bound.  This algorithm, FR-select, can be adapted to a distributed memory
parallel setting~\cite{San98a-e}.  FR-select picks the pivots based on sorting a
random sample~$S$ of~$\Oh{\sqrt{p}}$ objects, which can easily be done in
logarithmic time (e.g., the simple algorithm described in \cite{ABSS15} performs
brute-force comparisons of all possible pairs of sample elements).
Pivots~$\ell$ and~$r$ are chosen as the sample objects with
ranks~$k\card{S}/n\pm\Delta$ where~$\Delta=p^{1/4+\delta}$ for some small
constant~$\delta>0$.  For the analysis, one can choose~$\delta=1/6$, which
ensures that with high probability the range of possible values shrinks by a
factor $\Thsmall{p^{1/3}}$ in every level of recursion.

It is shown that a constant number of recursion levels suffice
if~$n=\Oh{p\log p}$ and if the objects are distributed randomly.  Note that the
latter assumption limits communication efficiency of the algorithm since it
requires moving all objects to random PEs for general inputs.

\renewcommand{\figurename}{Algorithm}
\begin{figure}
\begin{code}
\Funct{select}{\Declare{$s$}{Sequence \Of Object};
\Declare{$k$}{$\nat$}}{Object}\+\\
  \If $k=1$ \Then \Return $\min_is[1]@i$\RRem{base case}\\
  $(\ell,r)\Is \Id{pickPivots}(s)$\\
  $a \Is \seqGilt{e\in s}{e < \ell}$\\
  $b \Is \seqGilt{e\in s}{\ell \leq e \leq r}$\\
  $c \Is \seqGilt{e\in s}{e > r}$\RRem{partition}\\
  $n_a\Is \sum_{i}|a|@i$;\quad$n_b\Is \sum_{i}|b|@i$\RRem{reduction}\\
  \If $n_a  \geq k$ \Then \Return select$(a,k)$\\
  \If $n_a+n_b < k$ \Then \Return select$(c,k-n_a-n_b)$\\
  \Return select$(b,k-n_a)$
\end{code}
\caption{\label{fig:select}Communication efficient selection.}
\end{figure}

\paragraph*{Bulk Parallel Priority Queues.}\label{ss:ppq}
A natural way to parallelize priority queues is to use bulk
operations.  In particular, operation \deleteMinPar\ supports deletion
of the~$k$ smallest objects of the queue. Such a data structure can
be based on a heap where nodes store sorted sequences rather than objects~%
\cite{DeoPra92}. However, an even faster and simpler randomized way is to use
multiple sequential priority queues---one on each~PE~\cite{San98a-e}.
This data structure adopts the idea of Karp and Zhang~\cite{KarZha93}
to give every~PE a representative approximate view of the global situation
by sending inserted objects to random queues. However, in contrast to Karp
and Zhang~\cite{KarZha93}, \cite{San98a-e}~\nobreak implements an exact
\deleteMinPar\ using parallel selection and a few further tricks. Note
that the random insertions limit communication efficiency in this case.


\section{More Related Work}\label{s:related}

The fact that the amount of communication is important for parallel computing is
a widely studied issue. However, there are very few results on achieving sublinear
communication volume per processor.  In part, this is because models like parallel external
memory or resource-oblivious models assume that the input is located in a global
memory so that processing it requires loading all of it into the local caches at
least once. Refer to \cite{SSM13} for a more detailed discussion.

A good framework for studying bottleneck communication volume is the BSP model
\cite{Val94short} where a communication round with maximum amount of data per PE
$h$ (maximized over sent and received data) takes time $\ell+hg$ for some
machine parameters $\ell$ (latency) and $g$ (gap).  Summing the values of $h$
over all rounds yields the bottleneck communication volume we consider.  We do
not use the BSP model directly because we make heavy use of collective
communication routines which are not directly represented in the BSP model.
Further, the latency parameter $\ell$ is typically equal to $\alpha p$ in
real-world implementations of the BSP model (see also \cite{ABSS15}).  Also, we
are not aware of any concrete results in the BSP model on sublinear
communication volume in general or for top-$k$ problems in particular.  A
related recent line of research considers algorithms based on MapReduce (e.g.,
\cite{GSZ11}).  Communication volume is an important issue there but it seems
difficult to achieve sublinear communication volume since the MapReduce
primitive does not model locality well.

Recently, \emph{communication avoiding algorithms} have become an important
research direction for computations in high performance linear algebra and
related problems, e.g. \cite{CDKSY13}.  However, these results only apply to
nested loops with array accesses following a data independent affine pattern and
the proven bounds are most relevant for computations with work superlinear in
the input size.  We go into the opposite direction looking at problems with
linear or even sublinear work.

\paragraph*{Selection.}
Plaxton~\cite{Pla89} shows a superlogarithmic lower bound for selection on a
large class of interconnection networks. This bound does not apply to our model,
since we assume a more powerful network.

\paragraph*{Parallel Priority Queues.}
There has been intensive work on parallel priority queues.  The most scalable
solutions are our randomized priority queue~\cite{San98a-e} and Deo and Prasad's
parallel heap~\cite{DeoPra92}. Refer to~\cite{WGTT15TR} for a recent
over\-view of further approaches, most of which are for shared memory architectures
and mostly operate on centralized data structures with limited scalability.

\paragraph*{Multicriteria, Most Frequent Objects.}
Considerable work has been done on distributed top-$k$ computations in wide area
networks, sensor networks and for distributed streaming
models~\cite{BabOls03,CaoWang04,YiZhang13,HuaYiZha2012,KLEE2005}. However,
these papers use a master--worker approach where all communication has to go over a master node.
This implies up to $p$ times higher communication volume
compared to our results. Only rarely do randomized versions with better
running times exist, but nevertheless, communication volume at the master
node still increases with~$\sqrt{p}$\cite{HuaYiZha2012}.
Moreover, the scalability of TPUT~\cite{CaoWang04} and KLEE \cite{KLEE2005} is severely limited
by the requirement that the number of workers cannot exceed the number
of criteria. Furthermore, the top-$k$ most frequent objects problem
has received little attention in distributed streaming algorithms.
The only work we could find requires $\Oh{N p}$ working memory at the
master node, where~$N=\Oh{n}$ is the number of distinct objects in the
union of the streams~\cite{BabOls03}. The majority of papers
instead considers the significantly easier problem of identifying the
\emph{heavy hitters}, i.e. those objects whose occurrences account for
more than a fixed proportion of the input, or \emph{frequency tracking}, which tracks the
approximate frequencies of all items, but requires an additional selection step
to obtain the most frequent ones~ \cite{HuaYiZha2012,YiZhang13}.

\paragraph*{Sum Aggregation}
Much work has been done on aggregation in parallel and distributed settings,
e.g. in the database community~\cite{Cieslewicz2007,Li2013,Muller2015}.
However, these papers all ask for \emph{exact}
results for \emph{all} objects, not approximations for the~$k$ most important
ones.  We do not believe that exact queries are can be answered in a
communication-efficient manner, nor could we find any works that consider the
same problem, regardless of communication efficiency.

\fakepar{Data redistribution} using prefix sums is standard \cite{Ble89}. But
combining it with merging allows adaptive communication volume as described in
Section~\ref{s:redistribution}, which seems to be a new result. There are algorithms for
minimizing communication cost of data distribution in arbitrary networks.
However, these are much more expensive. Meyer auf der Heide et al. \cite{MOW96} use
maximum flow computations.  They also give an algorithm for meshes, which needs
more steps than the lower bound though. Solving the diffusion equations
minimizes the sum of the squares of the data to be moved in order to achieve
perfect load balance \cite{DieFM99}.  However, solving this equation is expensive
and for parallel computing we are more interested in the bottleneck
communication volume.

\section{Selection}\label{s:selection}

We consider the problem of identifying the elements of rank~$\leq k$ from a
set of~$n$ objects, distributed over~$p$ processors. First, we analyze
the classical problem with unsorted input, before turning to locally
sorted input in Sections~\ref{ss:musel} and~\ref{ss:approxMusel}.

\subsection{Unsorted Input}\label{ss:unsortedSelection}

Our first result is that using some minor adaptations and a more
careful analysis, the parallel FR-algorithm from
Section~\ref{s:prelim} does not actually need randomly distributed
data.
\begin{theorem}\label{thm:selectFR}
  Consider $n$ elements distributed over $p$ PE such that each PE holds
  $\Oh{n/p}$ elements.  The $k$ globally smallest of these elements can be
  identified in expected time
  $$\Oh{\frac{n}{p}+\Tword\min\left(\sqrt{p}\log_p n,
      \frac{n}{p}\right)+\Tstart\log n}\punkt$$
\end{theorem}
\begin{proof} (Outline)
The algorithm from~\cite{San98a-e} computes the pivots based on a sample~$S$ of size
$\Oh{\sqrt{p}}$. There, this is easy since the objects are randomly
distributed in all levels of recursion.  Here we can make no such assumption and
rather assume Bernoulli sampling with probability
$\sqrt{p}/\sum_i \card{s@i}$. Although this can be done in time proportional to the
local sample size, we have to be careful since (in some level of recursion)
the input might be so skewed that most samples come from the same
processor.  Since the result on sorting only works when the input data is
uniformly distributed over the~PEs, we have to account additional time
for spreading the samples over the~PEs.

Let~$x$ denote the~PE which maximizes $|s@x|$ in the current level of
recursion in the algorithm from Figure~\ref{fig:select}.  The total problem size
shrinks by a factor of $\Omsmall{p^{1/3}}$ in each level of recursion with high probability.  Hence, after
a constant number of recursion levels,~$x$ is $\Oh{n/p}$ as well.  Moreover, after~%
$\log_{\Omsmall{p^{1/3}}}\frac{n}{p}=\Ohsmall{\log_p\frac{n}{p}}$ further levels of
recursion, the problem is solved completely with high probability. Overall,
the total number of recursion levels is $\Ohsmall{1+\log_p\frac{n}{p}}=\Oh{\log_p{n}}$.

This already implies the claimed $\Oh{n/p}$ bound on internal computation
-- $\Oh{1}$ times $\Oh{n/p}$ until~$x$ has size $\Oh{n/p}$ and a geometrically
shrinking amount of work after that.

Of course, the same bound $\Ohsmall{\Tword\frac{n}{p}}$ also applies to the
communication volume. However, we also know that even if all samples come from a
single~PE, the communication volume in a single level of recursion is~$\Oh{\sqrt{p}}$.

Finally, the number of startups can be limited to $\Oh{\log p}$ per level of
recursion by using a collective scatter operation on those~PEs that have more than
one sample object. This yields $\Oh{\log p\log_pn}=\nobreak\Oh{\log n}$.
\end{proof}

\begin{corollary}
If $\Tstart$ and $\Tword$ are viewed as constants, the bound from Theorem~\ref{thm:selectFR}
reduces to $\Oh{\frac{n}{p}+\log p}$.
\end{corollary}

\noindent The bound from Theorem~\ref{thm:selectFR} immediately reduces
to~$\Oh{\frac{n}{p}+\log{n}}$ for constant~$\Tstart$ and~$\Tword$. We can
further simplify this by observing that when the~$\log{n}$ term
dominates~$n/p$, then~$\log{n} = \Oh{\log{p}}$.

\subsection{Locally Sorted Input}\label{ss:musel}

Selection on locally sorted input is easier than the unsorted problem from
Section~\ref{ss:unsortedSelection} since we only have to consider the locally
smallest objects and are able to locate keys in logarithmic time. Indeed, this
problem has been studied as the \emph{multisequence selection}
problem~\cite{VSIR91,SSP07}.

In~\cite{ABSS15}, we propose a particularly simple and intuitive method based on an
adaptation of the well-known quickselect algorithm~\cite{Hoa61b,MehSan08}.
For self-containedness, we also give that algorithm in Appendix~\ref{ss:msselect}, adapting it to the need in the present paper.
This algorithm needs running time~$\Oh{\Tstart\log^2kp}$ and has been on the
slides of Sanders' lecture on parallel algorithms since 2008 \cite{San08VLPA}.
Algorithm~\ref{alg:musel} can be viewed as a step backwards compared to our algorithm
from Section~\ref{ss:unsortedSelection} as using a single random pivot
forces us to do a deeper recursion. We do that because it makes it easy to
tolerate unbalanced input sizes in the deeper recursion levels. However, it is
possible to reduce the recursion depth to some extent by choosing pivots more
carefully and by using multiple pivots at once. We study this below for a variant
of the selection problem where we are flexible about the number $k$ of objects
to be selected.

\subsection[Flexible k, Locally Sorted Input]{Flexible $k$, Locally Sorted Input}\label{ss:approxMusel}
\newcommand{\kopt}{k^*}
\newcommand{\navg}{\bar{n}}
\newcommand{\kmax}{\overline{k}}
\newcommand{\kmin}{\underline{k}}

The $\Oh{\log^2pk}$ startups incurred in Section~\ref{ss:musel} can be reduced to~$\Oh{\log pk}$ if we are willing
to give up some control over the number of objects that are actually returned.  We now give
two input parameters~$\kmin$ and~$\kmax$ and the algorithm
returns the~$k$ smallest objects so that~$\kmin\leq k\leq \kmax$.
We begin with a simple algorithm that runs in logarithmic time if~%
$\kmax-\kmin=\Om{\kmax}$ and then explain how to
refine that for the case~$\kmax-\kmin=o(\kmax)$.

The basic idea for the simple algorithm is to take a Bernoulli sample
of the input using a success probability of~$1/x$, for~$x\in\kmin..\kmax$. Then, the
expected rank of the smallest sample object is~$x$, i.e., we have a
truthful estimator for an object with the desired rank.  Moreover, this object
can be computed efficiently if working with locally sorted data: the local rank
of the smallest local sample is geometrically distributed with parameter~%
$1/x$.  Such a number can be generated in constant time.
By computing the global minimum of these locally smallest samples,
we can get the globally smallest sample~$v$ in time~$\Oh{\Tstart\log p}$. We can
also count the exact number~$k$ of input objects  bounded by this estimate in time~%
$\Oh{\log k+\Tstart\log p}$---we locate~$v$ in each local data set in time~%
$\Oh{\log k}$ and then sum the found positions.  If~$k\in\kmin..\kmax$, we
are done. Otherwise, we can use the acquired information as in any variant of
quickselect.

At least in the recursive calls, it can happen that $\kmax$ is close to the
total input size~$n$.  Then it is a better strategy to use a dual algorithm
based on computing a global maximum of a Bernoulli sample.  In
Algorithm~\ref{alg:approxMusel} we give pseudocode for a combined algorithm
dynamically choosing between these two cases. It is an interesting problem which
value should be chosen for~$x$. The formula used in
Algorithm~\ref{alg:approxMusel} maximizes the probability that $k\in\kmin..\kmax$.
This value is close to the arithmetic mean when~$\kmax/\kmin\approx 1$ but it is
significantly smaller otherwise. The reason for this asymmetry is that larger
sampling rates decrease the variance of the geometric distribution.

\begin{theorem}\label{thm:approxMusel}
If $\kmax-\kmin=\Om{\kmax}$, then
Algorithm \Id{amsSelect} from Figure~\ref{alg:approxMusel} finds the~$k$
smallest elements with~$k\in\kmin..\kmax$ in expected time~%
$\Oh{\log\kmax+\Tstart\log p}$.
\end{theorem}
\begin{proof} (Outline)
One level of recursion takes time $\Oh{\Tstart\log p}$ for collective
communication operations~(min, max, or sum reduction) and time
$\Oh{\log \kmax}$ for locating the pivot $v$.
It remains to show that the expected recursion depth is constant.

We actually analyze a weaker algorithm that keeps retrying with the same parameters rather
than using recursion and that uses probe $x=\nobreak\kmin$. We show that, nevertheless,
there is a constant success probability (i.e.,
$\kmin\leq k\leq \kmax$ with constant probability).  The rank of $v$ is geometrically distributed with parameter
$1/x$.
The success probability becomes
$$\sum_{k=\kmin}^{\kmax}\left(1-\frac{1}{\kmin}\right)^{k-1} \frac{1}{\kmin}=
  \frac{\kmin}{\kmin-1} \left(1-\frac{1}{\kmin}\right)^{\kmin} - \left(1-\frac{1}{\kmin}\right)^{\kmax} \approx
  \frac{1}{e}-e^{-\frac{\kmax}{\kmin}}$$
which is a positive constant if~$\kmax-\kmin=\Om{\kmax}$.
\end{proof}

\begin{figure}[bt]
\begin{code}
\Funct{amsSelect}{\Declare{$s$}{Object\texttt{[]}};
\Declare{$\kmin$}{$\nat$}; \Declare{$\kmax$}{$\nat$}}{Object\texttt{[]}}\+\\
  \If $\kmin < n-\kmax$ \Then\RRem{min-based estimator}\+\\
    $x\Is \Id{geometricRandomDeviate}\left(1-\left(\frac{\kmin-1}{\kmax}\right)^{\frac{1}{\kmax-\kmin+1}}\right)$\\
    \If $x>|s|$ \Then $v\Is\infty$ \Else $v\Is s[x]$\\
    $v\Is \min_{1\leq i\leq p}v@i$\RRem{minimum reduction}\-\\
  \Else\RRem{max-based estimator}\+\\
    $x\Is \Id{geometricRandomDeviate}\left(1-\left(\frac{n-\kmax}{n-\kmin+1}\right)^{\frac{1}{\kmax-\kmin+1}}\right)$\\
    \If $x>|s|$ \Then $v\Is-\infty$ \Else $v\Is s[|s|-x+1]$\\
    $v\Is \max_{1\leq i\leq p}v@i$\RRem{maximum reduction}\-\\
  find $j$ such that $s[1..j] \leq v$ and $s[j+1..]> v$\\
  $k\Is \sum_{1\leq i\leq p}|j@i|$\RRem{sum all-reduction}\\
  \If $k<\kmin$ \Then\+ \\
    \Return $s[1..j]\cup\Id{amsSelect}(s[j+1..|s|], \kmin-k, \kmax-k, n-k)$\-\\
  \If $k>\kmax$ \Then\+\\
    \Return $\Id{amsSelect}(s[1..j], \kmin, \kmax, k)$\-\\
  \Return $s[1..j]$
\end{code}
\caption{\label{alg:approxMusel}Approximate multisequence
  selection.  Select the~$k$ globally smallest out of~$n$ objects
  with~$\kmin\leq\nobreak k\leq\nobreak
  \kmax$.  Function~\Id{geometricRandomDeviate} is a standard library
  function for generating geometrically distributed random variables
  in constant time \cite{PreEtAl92}.}
\end{figure}

\paragraph*{Multiple Concurrent Trials.} The running time of Algorithm~\ref{alg:approxMusel} is dominated by the
logarithmic number of startup overheads for the two reduction operations it
uses. We can exploit that reductions can process long vectors
using little additional time. The idea is to take~$d$ Bernoulli samples of the input and to compute
$d$ estimates for an object of rank~$x$. If any of these estimates turns out to have exact rank
between~$\kmin$ and~$\kmax$, the recursion can be stopped.
Otherwise, we solve a recursive instance
consisting of those objects enclosed by the largest underestimate and the smallest overestimate
found.

\begin{theorem}\label{thm:vectorMusel}
If $\kmax-\kmin=\Om{\kmax/d}$,
an algorithm processing batches of $d$ Bernoulli samples can be implemented to run in expected time
$\Oh{d\log\kmax+\Tword d +\Tstart\log p}$.
\end{theorem}
\begin{proof} (Outline)
A single level of recursion runs in time $\Oh{d\log\kmax+\Tword d + \Tstart\log p}$.
Analogous to the proof of Theorem~\ref{thm:approxMusel}, it can be shown that the
success probability is $\Om{1/d}$ for a single sample.  This implies that the
probability that any of the $d$ independent samples is successful is constant.
\end{proof}

For example, setting $d=\Th{\log p}$, we obtain communication time $\Oh{\Tstart\log p}$
and $\Oh{\log k\log p}$ time for internal work.

\section{Bulk Parallel Priority Queues}\label{s:pq}



We build a global bulk-parallel priority queue from
local sequential priority queues as in \cite{KarZha93,San98a-e}, but never
actually move elements.  This immediately implies that insertions simply go to
the local queue and thus require only~$\Oh{\log n}$ time without any communication.
Of course, this complicates operation \deleteMinPar. The number of
elements to be retrieved from the individual local queues can vary arbitrarily
and the set of elements stored locally is not at all representative for the
global content of the queue.  We can not even afford to individually remove the
objects output by \deleteMinPar\ from their local queues. 

We therefore replace the ordinary priority queues used in~\cite{San98a-e} by
search tree data structures that support insertion, deletion, selection,
ranking, splitting and concatenation of objects in logarithmic time (see also
Section~\ref{s:prelim}). To become independent of the actual tree size of up to
$n$, we furthermore augment the trees with two arrays storing the path to the
smallest and largest object respectively.  This way, all required operations
can be implemented to run in time $\Oh{\log k}$ rather than $\Oh{\log n}$.

Operation $\deleteMinPar$ now becomes very similar to the multi-sequence
selection algorithms from Section~\ref{ss:approxMusel} and Appendix~\ref{ss:msselect}.  The
only difference is that instead of sorted arrays, we are now working on search
trees.  This implies that selecting a local object with specified
local rank (during sampling) now takes time~$\Oh{\log k}$ rather than constant
time. However, asymptotically, this makes no difference since for any such
selection step, we also perform a ranking step, which takes time~$\Oh{\log k}$
anyway in both representations.

One way to implement the recursion in the selection algorithms is via splitting.
Since the split operation is destructive, after returning from the recursion, we have
to reassemble the previous state using concatenation.
Another way that might be faster and simpler in practice
is to represent a subsequence of~$s$ by~$s$ itself plus cursor information specifying
rank and key of the first and last object of the subsequence.

Now, we obtain the following by applying our results on selection from
Section~\ref{s:selection}.
\begin{theorem}
Operation~$\deleteMinPar$ can be implemented to run in the following
expected times. With fixed batch size~$k$, expected time $\Ohsmall{
\Tstart\log^2kp}$ suffices. For flexible batch size in~$\kmin..\kmax$,
where~$\kmax-\kmin=\Om{\kmax}$, we need expected time
$\Oh{\Tstart\log kp}$. If~$\kmax-\kmin=\Om{\kmax/d}$, expected time
$\Oh{d\log\kmax+\Tword d +\Tstart\log p}$ is sufficient.
\end{theorem}

Note that flexible batch sizes might be adequate for many applications. For example,
the parallel branch-and-bound algorithm from \cite{San98a-e} can easily be
adapted: In iteration~$i$ of its main loop, it deletes the smallest~$k_i=\Oh{p}$
elements (tree nodes) from the queue, expands these nodes in parallel, and
inserts newly generated elements (child nodes of the processed nodes).  Let
$K=\sum_ik_i$ denote the total number of nodes expanded by the parallel
algorithm.  One can easily generalize the proof from \cite{San98a-e} to show
that~$K=m+\Oh{hp}$ where~$m$ is the number of nodes expanded by a sequential
best first algorithm and~$h$ is the length of the path from the root to
the optimal solution. Also note that a typical branch-and-bound computation will
insert significantly more nodes than it removes---the remaining queue is discarded
after the optimal solutions are found. Hence, the local insertions of our
communication efficient queue are a big advantage over previous algorithms,
which move all nodes \cite{KarZha93,San98a-e}.

\section[Multicriteria Top-k]{Multicriteria Top-$k$}\label{s:multic}

In the sequential setting, we consider the following problem: Consider~$m$ lists~$L_i$
of scores that are sorted in decreasing order.  Overall relevance of an
object is determined by a scoring function~$t(x_1,\ldots,x_m)$ that is
monotonous in all its parameters.  For example, there could be~$m$ keywords for a
disjunctive query to a fulltext search engine, and for each pair of a keyword and an object,
it is known how relevant this keyword is for this object.  Many algorithms used
for this setting are based on Fagin's threshold algorithm~\cite{FLN03}.  The
lists are partially scanned and the algorithm maintains lower bounds for the
relevance of the scanned objects as well as upper bounds for the unscanned
objects. Bounds for a scanned object~$x$ can be tightened by retrieving a
score value for a dimension for which~$x$ has not been scanned yet. These
\emph{random accesses} are more expensive than scanning, in particular if the
lists are stored on disk.  Once the top-$k$ scanned objects have better lower
bounds than the best upper bound of an unscanned object, no more scanning is
necessary. For determining the exact relevance of the top-$k$ scanned objects,
further random accesses may be required.  Various strategies for scanning and
random access yield a variety of variants of the threshold
algorithm~\cite{BMSTW06}.

The original threshold algorithm works as
follows~\cite{FLN03}: In each of~$K$ iterations of the main loop, scan one object from each
list and determine its exact score using random accesses.
Let~$x_i$ denote the smallest score of a scanned object in~$L_i$.
Once at least~$k$ scanned objects have score at least~$t(x_1,\ldots,x_m)$, stop,
and output them.

We consider a distributed setting where each~PE has a subset of the objects and~$m$ sorted
lists ranking its locally present objects.
We describe communication efficient distributed
algorithms that approximate the original threshold algorithm~(TA)~\cite{FLN03}.
First, we give a simple algorithm assuming random data distribution of the input objects
(RDTA) and then describe a more complex algorithm for arbitrary data distribution~(DTA).

\paragraph*{Random Data Distribution.}
Since the data placement is independent of the relevance of the objects, the
top-$k$ objects are randomly distributed over the~PEs. Well known
balls-into-bins bounds give us tight high probability bounds on the maximal
number~$\kmax$ of top-$k$ objects on each~PE~\cite{RaaSte98}. Here, we work with
the simple bound~$\kmax=\Ohsmall{\sfrac{k}{p}+\log p}$.  The RDTA algorithm
simply runs~TA locally to retrieve the~$\kmax$ locally most relevant objects on
each~PE as result candidates.  It then computes a global threshold as the
maximum of the local thresholds and verifies whether at least~$k$ candidate
objects are above this global threshold.  In the positive case, the~$k$ most
relevant candidates are found using the selection algorithm
from~\cite{San98a-e}.  Otherwise,~$\kmax$ is increased and the algorithm is
restarted.  In these subsequent iterations,~PEs whose local threshold is worse
than the relevance of the~$k$-th best element seen so far do not need to
continue scanning. Hence, we can also trade less local work for more rounds of
communication by deliberately working with an insufficiently small value of~$\kmax$.

\paragraph*{Arbitrary Data Distribution.}

We give pseudocode for Algorithm DTA in Figure~\ref{alg:dta}.

\def\tmin{t_\text{min}}
\begin{figure}
\begin{code}
\Procedure DTA$(\seq{L_1,\ldots,L_m}, t, k)$\+\\
  $K\Is \ceil{\frac{k}{mp}}$\RRem{lower bound for $K$}\\
  \Do\+\\
    \For $i \Is 1 ~\To m$ \Do\+\\
      $L'_i$ \Is \Id{amsSelect($L_i, K, 2K, \card{L_i}$)},~ $(\cdot,x_i) \Is L'_i$.last\RRem{$x_i$ min selected score}\-\\
    $\tmin \Is t(x_1,\ldots,x_m)$\RRem{threshold}\\
    \For $i \Is 1 ~\To m$\RRem{each PE locally}\+\\
      $R_i, H_i \Is 0$\\
      \For $j \Is 1 ~\To y=\Oh{\log{K}}$\+\\
        sample an element $x=(o,s)$ from $L'_i$\\
        \If $\exists_{j<i,s'}:{(o,s') \in L'_j}$ \Then $R_i$++\RRem{ignore samples contained in other $L'_j$}\\
        \Elsif $t(o) \geq \tmin$ \Then $H_i$++\RRem{$t(o)$ does lookups in all $L_i$ to compute score}\-\\
      $\ell_i \Is \card{L'_i}\left(1-\frac{R_i}{y}\right)\cdot\frac{H_i}{y}$\RRem{truthful estimator of \#hits for~$L'_i$ on this PE}\-\\
    $H \Is \sum_{j=1}^{p}\left(\sum_{i=1}^{m}\ell_i\right)@j$\\
    $K \Is 2K$\RRem{exponential search}
    \-\\
  \Until $H\geq 2k$\RRem{$\Rightarrow\geq k$ hits found whp }\\
  \Return $(\tmin,\seq{L'_1,\ldots,L'_m})$
\end{code}
\caption{Multicriteria Top-$k$ for arbitrary data distribution (DTA) on Lists
  $L_1,\ldots,L_m$ with scoring function~$t$}\label{alg:dta}
\end{figure}

\begin{theorem}
Algorithm DTA requires expected time $\Oh{m^2\log^2 K +\Tword m\log K
  +\Tstart\log p\log K}$ to identify a set of at most $\Oh{K}$ objects
that contains the set of $K$ objects scanned by TA.
\end{theorem}
\begin{proof} (Outline)
Algorithm DTA ``guesses'' the number~$K$ of list rows scanned by TA
using exponential search. This yields $\Oh{\log K}$ \emph{rounds} of
DTA.\@ In each round, the approximate multisequence selection
algorithm from Section~\ref{ss:approxMusel} is used to approximate the
globally~$K$-th largest score~$x_i$ in each list. Define~$L'_i$ as the
prefix of~$L_i$ containing all objects of score at least~$x_i$. This takes
time~$\Oh{\log K+\Tstart\log p}$ in expectation by
Theorem~\ref{thm:approxMusel}. Accumulating the searches for all~$m$
lists yields~$\Oh{m\log K +\Tword m +\Tstart\log p}$ time.  Call an
object selected by the selection algorithm a \emph{hit} if its
relevance is above the threshold~$t(x_1,\ldots,x_m)$.  DTA estimates
the number of hits using sampling.  For each PE and list~$i$
separately,~$y=\Oh{\log K}$ objects are sampled from $L'_i$. Each sample takes
time $\Oh{m}$ for evaluating~$t(\cdot)$.  Multiplying this with $m$ lists and
$\log K$ iterations gives a work bound of $m^2\log K$.  Note that this
is sublinear in the work done by the sequential algorithm which takes
time~$m^2K$ to scan the~$K$ most important objects in every list.  To
eliminate bias for objects selected in multiple lists $L'_j$, DTA only counts
an object~$x$ if it is sampled in the first list that contains it.
Otherwise,~$x$ is rejected. DTA also counts the number of
rejected samples~$R$. Let~$H$ denote the number of nonrejected hits in
the sample. Then,~$\ell'_i \Is\card{L'_i}(1-\nobreak R/y)$ is a truthful
estimate for the length of the list with eliminated duplicates
and~$\frac{H}{y}\ell'_i$ is a truthful estimate for the number of hits
for the considered list and~PE.\@ Summing these~$mp$ values using a
reduction operation yields a truthful estimate for the overall number
of hits. DTA stops once this estimate is large enough such that with
high probability the actual number of hits is at least~$k$.  The
output of DTA are the prefixes $L'_1,\ldots,L'_m$ of the lists on each PE, the union of
which contains the~$k$ most relevant objects with high
probability. It also outputs the threshold $t(x_1,\ldots,x_m)$. In
total, all of this combined takes expected time $\Ohsmall{m^2\log^2 K
  +\Tword m\log K +\Tstart\log p\log K}$.
\end{proof}

Actually computing the~$k$ most frequent objects amounts to scanning the result lists to
find all hits and, if desired, running a selection algorithm to identify the~$k$
most relevant among them. The scanning step may involve some load imbalance, as in the
worst case, all hits are concentrated on a single~PE. This seems unavoidable
unless one requires random distribution of the objects. However, in practice it may
be possible to equalize the imbalance over a large number of concurrently
executed queries.

\paragraph*{Refinements.}
We can further reduce the latency of DTA by trying several values of $K$ in each
iteration of algorithm DTA.\@ Since this involves access to only few
objects, the overhead in internal work will be limited. In practice, it may also
be possible to make good guesses on $K$ based on previous executions of the
algorithm.

\section[Top-k Most Frequent Objects]{Top-$k$ Most Frequent Objects}\label{s:topk}

\newcommand{\PAC}{\emph{PAC}\xspace}
\newcommand{\EC}{\emph{EC}\xspace}
\newcommand{\PEC}{\emph{PEC}\xspace}
\newcommand{\Naive}{\emph{Naive}\xspace}
\newcommand{\NaiveTree}{\emph{Naive Tree}\xspace}

We describe two \textit{probably approximately correct (PAC)} algorithms to
compute the top-$k$ most frequent objects of a multiset~$M$ with~$\card{M}=n$,
followed by a \textit{probably exactly correct (PEC)} algorithm for suitable
inputs.  Sublinear communication is achieved by transmitting only a small random
sample of the input.  For bound readability, we assume that~$M$ is distributed
over the~$p$ PEs so that none has more than~$\Oh{n/p}$ objects.  This is not
restricting, as the algorithms' running times scale linearly with the maximum
fraction of the input concentrated at one~PE.

We express the algorithms' error relative to the total input size.  This is
reasonable---consider a large input where~$k$ objects occur twice, and all
others once.  If we expressed the error relative to the objects' frequencies,
this input would be infeasible without communicating all elements.  Thus, we
refer to the algorithms' relative error as~$\Error$, defined so that the
absolute error~$\Error n$ is the count of the most frequent object that was not
output minus that of the least frequent object that was output, or 0 if the
result was exact.  Let $\delta$ limit the probability that the algorithms
exceeds bound~$\eps$, i.e.  $\prob{\Error>\eps} < \delta$. We refer to the
result as an~$(\eps,\delta)$-approximation.

\subsection{Basic Approximation Algorithm}\label{ss:topkapprox}

\renewcommand{\figurename}{Figure}
\newcommand{\m}[1]{\textcolor{red}{\textbf{#1}}}
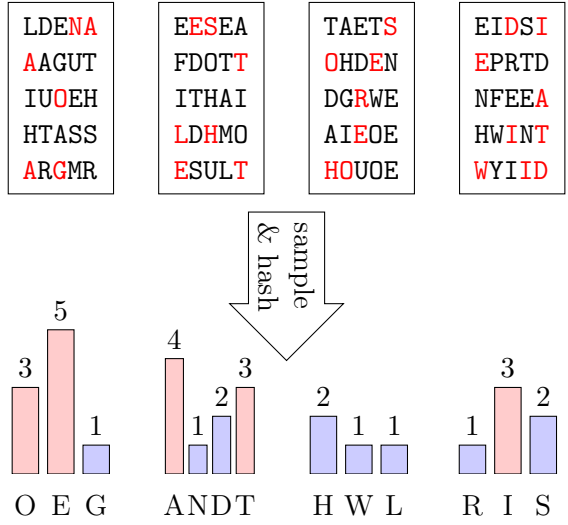
\begin{figure}
\center
\begin{tikzpicture}[
    >=stealth,
    node distance=3cm,
    input/.style={
      rectangle,
      shape border rotate=90,
      text width=1cm,
      draw,
      inner sep=2mm
    }
  ]
  \node[input] (i1) at (0,5) {\texttt{LDE\m{NA}\\\m{A}AGUT\\IU\m{O}EH\\HTASS\\\m{A}R\m{G}MR}};
  \node[input] (i2) at (2,5) {\texttt{E\m{ES}EA\\FDOT\m{T}\\ITHAI\\\m{L}D\m{H}MO\\\m{E}SUL\m{T}}};
  \node[input] (i3) at (4,5) {\texttt{TAET\m{S}\\\m{O}HD\m{E}N\\DG\m{R}WE\\AI\m{E}OE\\\m{HO}UOE}};
  \node[input] (i4) at (6,5) {\texttt{EI\m{D}S\m{I}\\\m{E}PRTD\\NFEE\m{A}\\HW\m{I}N\m{T}\\\m{W}YI\m{ID}}};

  \node[draw=black,single arrow,fill=none,rotate=270,text width=1.2cm] at (3,2.75) {sample\\\& hash};

\pgfplotsset{ybar, ymin=0, ymax=5, width=2.9cm, height=3.5cm,%
  axis line style= { draw opacity=0 }, ytick=\empty,%
  axis x line*=bottom, every tick/.style={opacity=0},%
  enlarge x limits=0.2, nodes near coords,%
};

\begin{axis}[at={(-0.04\linewidth,0)}, symbolic x coords={O,E,G}]
\addplot[draw=black,fill=blue!20, bar shift=0mm] coordinates {(G, 1)};
\addplot[draw=black,fill=red!20, bar shift=0mm, postaction={pattern=north east lines, pattern color=red!80!black}] coordinates {(O,3) (E, 5)};
\end{axis}

\begin{axis}[at={(0.08\linewidth,0)}, symbolic x coords={A,N,D,T}, bar width=2.4mm]
\addplot[draw=black,fill=blue!20, bar shift=0mm] coordinates { (N, 1) (D, 2) };
\addplot[draw=black,fill=red!20, bar shift=0mm, postaction={pattern=north east lines, pattern color=red!80!black}] coordinates {(A, 4) (T, 3)};
\end{axis}

\begin{axis}[at={(0.2\linewidth,0)}, symbolic x coords={H,W,L}]
\addplot[draw=black,fill=blue!20] coordinates {(H, 2) (L,1) (W,1)};
\end{axis}

\begin{axis}[at={(0.32\linewidth,0)}, symbolic x coords={R,I,S}]
\addplot[draw=black,fill=blue!20, bar shift=0mm] coordinates {(R, 1) (S, 2)};
\addplot[draw=black,fill=red!20, bar shift=0mm, postaction={pattern=north east lines, pattern color=red!80!black}] coordinates {(I, 3)};
\end{axis}

\end{tikzpicture}
\caption{Example for the PAC top-k most frequent objects algorithm
from Section~\ref{ss:topkapprox}. Top: input distributed over 4~PEs,
sampled elements ($\rho=0.3$) highlighted in red. Bottom: distributed
hash table counting the samples. The $k=5$ most frequent objects,
highlighted in red, are to be returned after scaling with $1/\rho$.
Estimated result: $(E, 17), (A, 13), (T, 10), (I, 10), (O, 10)$. Exact
result: $(E,16), (A,10), (T, 10), (I, 9), (D, 8)$. Object $D$ was
missed, instead $O$ (count $7$) was returned. Thus, the algorithm's error
is $8-7=1$ here.}\label{fig:sample}
\end{figure}

First, we take a Bernoulli sample of the input.  Sampling is done locally.  The
frequencies of the sampled objects are counted using distributed hashing---a
local object count with key~$k$ is sent to PE~$h(k)$ for a hash function~$h$
that we here expect to behave like a random function.  We then select the~$k$
most frequently sampled objects using the unsorted selection algorithm from
Section~\ref{ss:unsortedSelection}. An example is illustrated
in Figure~\ref{fig:sample}.

\begin{theorem}\label{thm:topkpac}
  Algorithm \PAC can be implemented to compute an~$(\eps,\delta)$-approximation
  of the top-$k$ most frequent objects in expected time
  $\Ohsmall{\Tword\frac{\log{p}}{p \eps^2}\log{\frac{k}{\delta}} + \Tstart\log n}$.
\end{theorem}

\begin{lemma}\label{lem:topkapproxtime}
  For sampling probability~$\rho$, Algorithm \PAC runs in
  $\Oh{\frac{n\rho}{p} + \Tword\frac{n\rho}{p}\log{p} + \Tstart\log n}$ time in
  expectation.
\end{lemma}
\begin{proof} (Outline)
  Bernoulli sampling is done in expected time $\Oh{n\rho/p}$ by generating skip
  values with a geometric distribution using success probability~$\rho$.  Since
  the number of sample elements in a Bernoulli sample is a random variable, so
  is the running time.  To count the sampled objects, each PE aggregates its
  local samples in a hash table, counting the occurrence of each sample object
  during the sampling process.  It inserts its share of the sample into a
  distributed hash table~\cite{SSM13} whose hash function we assume to behave
  like a random function, thus distributing the objects randomly among the PEs.
  The elements are communicated using indirect delivery to maintain logarithmic
  latency.  This requires
  $\Ohsmall{\Tword\frac{n\rho}{p}\log{p} + \Tstart\log{p}}$ time in expectation.
  To minimize worst-case communication volume, the incoming sample counts are
  merged with a hash table in each step of the reduction.  Thus, each PE
  receives at most one message per object assigned to it by the hash function.

  From this hash table, we select the object with rank~$k$ using Algorithm
  \ref{fig:select} in expected time
  $\Ohsmall{\frac{n\rho}{p} + \Tword\frac{n\rho}{p} + \Tstart\log{n\rho}}$.
  This pivot is broadcast to all PEs, which then determine their subset of at
  least as frequent sample objects in expected time
  $\Ohsmall{\frac{n\rho}{p} + \Tstart\log{p}}$.  These elements are returned.
  Overall, the claimed time complexity follows using the estimate~$p\leq n$
  and~$n\rho\leq n$ in order to simplify the~$\Tstart$-term.
\end{proof}

\begin{lemma}\label{lem:topkapproxerr}
  Let~$\rho \in (0,1)$ be the sampling probability.  Then, Algorithm \PAC's
  error~$\Error$ exceeds a value of~$\Delta/n$ with probability
  \[
    \prob{\Error n \geq \Delta} \leq 2ne^{-\frac{\Delta^2\rho k}{12 n}} +
    ke^{-\frac{\Delta^2 \rho}{8 n}}.
  \]
\end{lemma}

\begin{proof}
  We bound the error probability as follows: the probability that the
  error~$\Error$ exceeds some value~$\Delta/n$ is at most~$n$ times the
  probability that a single object's value estimate deviates from its true value
  by more than~$\Delta/2$ in either direction.  This probability can be bounded
  using Chernoff bounds.  We denote the count of element~$j$ in the input
  by~$x_j$ and in the sample by~$s_j$.  Further, let~$F_j$ be the probability
  that the error for element~$j$ exceeds~$\Delta/2$ in either direction.
  Let~$j \geq k$, and observe that~$x_j \leq n/k$.  Using
  Equations~\ref{eq:ch_leq} and~\ref{eq:ch_geq}
  with~$X=s_j$,~$\expect[X]=\rho x_j$, and~$\varphi=\frac{\Delta}{2 x_j}$, we
  obtain:
  \begin{align*}
    F_j &\leq
          \prob{s_j \leq \rho \left(x_j - \frac{\Delta}{2}\right)} \!+\!
          \prob{s_j \geq \rho \left(x_j + \frac{\Delta}{2}\right)} \\
        &\leq
          e^{-\frac{\Delta^2 \rho k}{8 n}} + e^{-\frac{\Delta^2 \rho k}{12 n}}
          \leq 2e^{-\frac{\Delta^2 \rho k}{12n}}\punkt
  \end{align*}

  This leaves us with the most frequent~$k-1$ elements, whose counts can be
  bounded as~$x_j \leq n$.  As overestimating them is not a concern, we apply
  the Chernoff bound in Equation~\ref{eq:ch_leq} and
  obtain~$\prob{s_j \leq \rho \left(x_j - \frac{\Delta}{2}\right)} \leq
  e^{-\frac{\Delta^2\rho}{8n}}$.
  In sum, all error probabilities add up to the claimed value.
\end{proof}

We bound this result by an error probability~$\delta$.  This allows us to
calculate the minimum required sample size given~$\delta$ and~$\Delta = \eps n$.
Solving the above equation for~$\rho$ yields
\begin{equation}
  \rho n \geq
  \frac{4}{\eps^2}\cdot\max{\left(
      \frac{3}{k}\ln{\frac{4n}{\delta}},~ 2\ln{\frac{2k}{\delta}}
    \right)},\label{eq:sampleSize}
\end{equation}
which is dominated by the latter term in most cases and
yields~$\rho n \geq \frac{8}{\eps^2} \ln{\frac{2k}{\delta}}$ for the expected
sample size.

\begin{proof}\emph{(Theorem~\ref{thm:topkpac})}
  Equation~\ref{eq:sampleSize} yields
  $\rho n = \smash{\Oh{\frac{1}{\eps^2} \log{\frac{k}{\delta}}}}$.  The claimed
  running time bound then follows from Lemma~\ref{lem:topkapproxtime}.
\end{proof}

Note that if we can afford to aggregate the local input, we can also use the Sum
Aggregation algorithm from Section~\ref{s:sumagg} and associate a value of 1
with each object.

\subsection{Increasing Communication Efficiency}\label{ss:topkcount}
Sample sizes proportional to~$1/\eps^2$ quickly become unacceptably large
as~$\eps$ decreases.  To remedy this, we iterate over the local input a second
time and count the most frequently sampled objects' occurrences exactly. This
allows us to reduce the sample size and improve communication efficiency at the
cost of increased local computation.  We call this \emph{Algorithm EC} for \emph{exact counting}.
Again, we begin by taking a Bernoulli sample.  Then we find the~$\kopt\geq k$
globally most frequent objects in the sample using the unsorted selection
algorithm from Section~\ref{ss:unsortedSelection}, and count their frequency in
the overall input exactly.  The identity of these objects is broadcast to all
PEs using an all-gather (gossiping, all-to-all broadcast) collective
communication operation.  After local counting, a global reduction sums up the
local counts to exact global values.  The~$k$ most frequent of these are then
returned.

\begin{lemma}\label{lem:topkec}
  When counting the~$k'$ most frequently sampled objects' occurrences exactly,
  a sample size of~$n\rho = \frac{2}{\eps^2 k'}\ln\frac{n}{\delta}$ suffices to
  ensure that the result of Algorithm \EC is an~$(\eps,\delta)$-approximation of
  the top-$k$ most frequent objects.
\end{lemma}
\begin{proof}
  With the given error bounds, we can derive the required sampling
  probability~$\rho$ similar to Lemma~\ref{lem:topkapproxerr}. However, we need
  not consider overestimation of the~$\kopt$ most frequent objects, as their
  counts are known exactly.  We can also allow the full margin of error towards
  underestimating their frequency ($\Delta$ instead of $\Delta/2$) and can
  ignore overestimation.  This way, we obtain a total expected sample size
  of~$\rho n \geq \frac{2}{\eps^2 \kopt} \ln{\frac{n}{\delta}}$.
\end{proof}

We can now calculate the value of~$k'$ that minimizes the total communication
volume and obtain
$\kopt =\nobreak \max(k, \frac{1}{\eps} \sqrt{\frac{2\log{p}}{p}
  \ln{\frac{n}{\delta}}}\,)$.
Substituting this into the sample size equation of Lemma~\ref{lem:topkpec} and
adapting the running time bound of Lemma~\ref{lem:topkapproxtime} then yields
the following theorem.

\begin{theorem}\label{thm:topkcount}
  Algorithm \EC can be implemented to compute an~$(\eps,\delta)$-approximation of
  the top-$k$ most frequent objects in
  $\Oh{\frac{n}{p} + \Tword \frac{1}{\eps} \sqrt{\frac{\log{p}}{p} \cdot
      \log{\frac{n}{\delta}}} + \Tstart\log{n}}$ time in expectation.
\end{theorem}
\begin{proof}
  Sampling and hashing are done as in Algorithm \PAC
  (Section~\ref{ss:topkapprox}). We select the object with rank
  $\kopt =\nobreak \max(k, \frac{1}{\eps} \sqrt{\frac{2\log{p}}{p}
    \ln{\frac{n}{\delta}}}\,)$
  as a pivot.  This requires expected time
  $\Ohsmall{\frac{n\rho}{p} + \Tword\frac{n\rho}{p} + \Tstart\log{n\rho}}$ using
  Algorithm~\ref{fig:select}.  The pivot is broadcast to all PEs, which then
  determine their subset of at least as frequent sample objects using
  $\Ohsmall{\sfrac{n\rho}{p}+\Tstart\log{p}}$ time in expectation.  Next,
  these~$\kopt$ most frequently sampled objects are distributed to all PEs using
  an all-gather operation in time $\Ohsmall{\Tword \kopt+\Tstart\log{p}}$.  Now,
  the PEs count the received objects' occurrences in their local input, which
  takes~$\Oh{n/p}$ time.  These counts are summed up using a vector-valued
  reduction, again requiring~$\Ohsmall{\Tword \kopt + \Tstart \log{p}}$ time.
  We then apply Algorithm~\ref{fig:select} a second time to determine the~$k$
  most frequent of these objects.  Overall, the claimed time complexity follows
  by substituting~$\kopt$ for~$k'$ in the sampling probability from
  Lemma~\ref{lem:topkec}.
\end{proof}

Substituting~$\kopt$ from before then yields a total required sample size
of~$\rho n = \frac{1}{\eps} \cdot
\sqrt{\sfrac{p}{\log{p}}\ln{\sfrac{n}{\delta}}}$
for Algorithm \EC.  Note that this term grows with~$\frac{1}{\eps}$ instead
of~$1/\eps^2$, reducing per-PE communication volume to
$\Oh{\frac{1}{\eps}\sqrt{\frac{\log{p}}{p}\log{\frac{n}{\delta}}}}$ words.

To continue the example from Figure~\ref{fig:sample}, we may
set~$\kopt=8$. Then, the $\kopt$ most frequently sampled objects $(E,
A, T, I, O, D, H,S)$ with $(16, 10, 10, 9, 7, 8, 7, 6)$
occurrences, respectively, will be counted exactly. The result would
now be correct.

\begin{figure}
\center
\begin{tikzpicture}
  \draw[->] (0,0) -- (4.4,0) node[right] {$i$};
  \draw[->] (0,0) -- (0,3.6) node[left] {$x_i$};
  \draw (0,3) to[out=-10,in=160] (1,2)%
    to[out=-30,in=120] (1.6,1.1)%
    to[out=-60,in=150] (2.5,0.5)%
    to[out=-30,in=170] (4.2,0.1);
  \draw[dashed] (0,2.5) to[out=-10,in=160] (1,1.6)%
    to[out=-30,in=140] (1.6,0.8)%
    to[out=-40,in=150] (2,0.55)%
    node[below] {$\underline{\hat{x}_i}$}
    to[out=-35,in=175] (4.2,0.05);
  \draw[dashed] (0,3.5) to [out=-10,in=140] (1,2.4)%
    to[out=-30,in=110] (1.6,1.65)%
    to[out=-60,in=150] (2,1)%
    to[out=-28,in=174] node[midway,above] {$\overline{\hat{x}_i}$} (4.2,0.18);
  \draw[-] (1,0) -- (1,3) node[above] {$k$};
  \draw[-] (1.625,0) -- (1.625,3) node[above] {$\kopt$};
  \draw[very thick,line cap=round,dash pattern=on 0pt off 2.5pt] (1,2) -- (3.5,2) node[right] {$x_k$};
  \draw[very thick,line cap=round,dash pattern=on 0pt off 2.5pt] (1.625,1.07) -- (3.5,1.07) node[right] {$x_{\kopt}$};
  \draw[<->,thick] (3,2) --  node[right] {$\Delta$} (3,1.07);
  \draw[|-|,line width=0.75pt,line cap=round,dash pattern=on 0pt off 1.5pt] (0.99,1.6) -- (1.635,1.6);
\end{tikzpicture}

\caption{Example of a distribution with a gap. The dashed lines
indicate the upper and lower high-probability bounds on $x_i$
estimated from the sample}\label{fig:gap}
\end{figure}
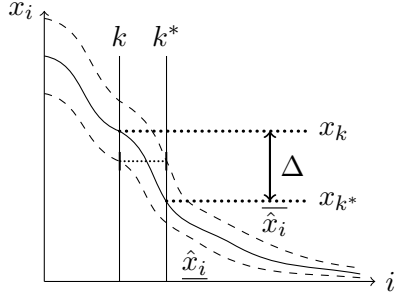

Note that Algorithm \EC can be \emph{less} communication efficient than Algorithm
\PAC if~$\eps$ is large, i.e. the result is very approximate.  Then,~$\kopt$ can
be prohibitively large, and the necessity to communicate the identity of the
objects to be counted exactly, requiring
time~$\Oh{\Tword \kopt + \Tstart \log{p}}$, can cause a loss in communication
efficiency.

\subsection{Probably Exactly Correct Algorithm}\label{ss:topkpec}

If any significant \textit{gap} exists in the frequency distribution of the
objects (see Figure~\ref{fig:gap} for an example), we perform exact counting on all likely relevant objects, determined
from their sample count.  Thus,
choose~$\kopt$ to ensure that the top-$k$ most frequent objects in the input are
among the top-$\kopt$ most frequent objects in the sample with probability at
least~$1-\delta$.

A \emph{probably exactly correct (PEC)} algorithm to compute the top-$k$ most
frequent objects of a multiset whose frequency distribution is sufficiently
sloped can therefore be formulated as follows.  Take a small sample with
sampling probability~$\rho_0$, the value of which we will consider later.  From
this small sample, we deduce the required value of~$\kopt$ to fulfill the above
requirements.  Now, we apply Algorithm \EC using this value of~$\kopt$.
Let~$x_i$ be the object of rank~$i$ in the input, and~$\hat s_j$ that of
rank~$j$ in the small sample.

\begin{lemma}\label{lem:topkpec}
  It suffices to choose~$\kopt$ in such a way that
  $\hat s_{\kopt} \leq \rho_0 x_k - \sqrt{2\rho_0 x_k \ln{\frac{k}{\delta}}} =
  \expect{\hat s_k} - \sqrt{2\expect{\hat s_k}\ln{\frac{k}{\delta}}}$
  to ensure correctness of Algorithm \PEC's result with
  probability at least $1-\delta$.
\end{lemma}
\begin{proof}
  We use the Chernoff bound from Equation~\ref{eq:ch_leq} to bound the
  probability of a top-$k$ object not being among the top-$\kopt$ objects of the
  sample.  Define~$\hat s(x_i)$ as the number of samples of the object with
  input rank~$i$ in the first sample, and~$x(\hat s_j)$ to be the exact number
  of occurrences in the input of the object with rank~$j$ in the first sample.
  Using~$X=\hat s(x_k)$, $\expect[X]=\rho_0 x_k$,
  and~$\varphi=1-\frac{\hat s_{\kopt}}{\rho_0 x_k}$, we obtain
  \begin{align*}
    \sum_{j=1}^{k}{\prob{\hat s(x_j) \leq \hat s_{\kopt}}}
    &\leq k\cdot\prob{\hat s(x_k)\leq \hat s_{\kopt}}\\
    &\leq k\cdot e^{-\left(1-\frac{\hat s_{\kopt}}{\rho_0 x_k}\right)^2\frac{\rho_0 x_k}{2}}\punkt
  \end{align*}
  We bound this value by the algorithm's error probability~$\delta$ and solve
  for~$\hat s_{\kopt}$, which yields the claimed value.
\end{proof}

This only works for sufficiently sloped distributions, as
otherwise~$\kopt \gg k$ would be necessary.  Furthermore, it is clear that the
choice of~$\rho_0$ presents a trade--off between the time and communication
spent on the first sample and the exactness of~$\kopt$, which we have to
estimate more conservatively if the first sample is small.  This is due to less
precise estimations of~$\expect{\hat s_k}=\rho_0 x_k$ if~$\rho_0$ is small.  To
keep things simple, we can choose a relative error bound~$\eps_0$ and use the
sample size from our PAC algorithm of Theorem~\ref{thm:topkpac}.  The value
of~$\eps_0$---and thus~$\rho_0$---that minimizes the communication volume
depends on the distribution of the input data.

\begin{theorem}\label{thm:topkpectime}
  If the value of~$\kopt$ computed from Lemma~\ref{lem:topkpec} satisfies
  $\kopt=\Oh{k}$, then Algorithm \PEC requires time asymptotically equal to the
  sum of the running times of algorithms \PAC and \EC from
  Theorems~\ref{thm:topkpac} and~\ref{thm:topkcount}.
\end{theorem}
\begin{proof}\emph{(Outline)}
  In the first sampling step, we are free to choose an arbitrary relative error
  tolerance~$\eps_0$.  The running time of this stage is
  $\Ohsmall{\Tword\frac{\log{p}}{p \eps_0^2}\log{\frac{n}{\delta}} + \Tstart\log
    n}$
  by Theorem~\ref{thm:topkpac}.  We then estimate~$\kopt$ by substituting the
  high-probability bound
  $\expect[\hat s_k] \geq \hat s_k-\sqrt{2\hat s_k \ln(1/\delta)}$ (whp) for its
  expected value in Lemma~\ref{lem:topkpec} (note that as~$\hat s_k$ increases
  with growing sample size and thus grows with falling~$\eps_0$, the precision
  of this bound increases).  In the second stage, we can calculate the required
  value of~$\eps$ from~$\kopt$ by solving the expression for~$\kopt$ in the
  proof of Theorem~\ref{thm:topkcount} for~$\eps$, and obtain
  $\eps = \frac{1}{\kopt} \sqrt{\frac{2\log{p}}{p}
    \log{\frac{n}{\delta}}}$.
  Since~$\kopt=\Ohsmall{k}$, the second stage's running time is as in
  Theorem~\ref{thm:topkcount}.  In sum, the algorithm requires the claimed
  running time.
\end{proof}

\paragraph*{Zipf's Law}\label{sss:zipf} In its simplest form, Zipf's Law
states that the frequency of an object from a multiset~$M$ with~$|M|=n$ is
inversely proportional to its rank among the objects ordered by frequency. Here,
we consider the general case with exponent parameter~$s$, i.e.~$x_i=
n\frac{i^{-s}}{H_{n,s}}$, where~$H_{n,s}\!=\sum_{i=1}^{n}{i^{-s}}$ is the~$n$-th
\textit{generalized harmonic number}.

\begin{theorem}\label{thm:zipf}
  For inputs distributed according to Zipf's law with exponent~$s$, a sample
  size of $\rho n=\nobreak 4 k^s H_{n,s} \ln{\frac{k}{\delta}}$ is sufficient to
  compute a probably exactly correct result.  Algorithm \PEC then requires
  expected time
  $\Oh{\frac{n}{p} + \Tword\frac{\log{p}}{p}k^s H_{n,s} \log{\frac{k}{\delta}} +
    \Tstart\log{n}}$
  to compute the~$k$ most frequent objects with probability at least $1-\delta$.
\end{theorem}
\begin{proof}
  Knowing the value distribution, we do not need to take the first sample.
  Instead, we can calculate the expected value of~$\kopt$ directly from the
  proof of Lemma~\ref{lem:topkpec} and obtain
  \[
    \expect{\kopt} = k\left( 1 - \sqrt{\frac{2\ln{\frac{k}{\delta}}}{\rho x_k}}
    \right)^{-\frac{1}{s}}\punkt
  \]
  This immediately yields~$\rho > \frac{2}{x_k}\ln{\frac{k}{\delta}}$, and we
  choose~$\rho = \frac{4}{x_k}\ln{\frac{k}{\delta}} = 4 \frac{k^sH_{n,s}}{n}
  \ln{\frac{k}{\delta}}$,
  i.e. twice the required minimum value.  This gives us the claimed sample size.
  Now, we obtain
  $\expect{\kopt} = \smash{k \left(2 + \sqrt{2}\right)} ^{\frac{1}{s}} \approx
  \sqrt[s]{3.41}k$.
  In particular,~$\kopt$ is only a constant factor away from $k$.  Plugging the
  above sampling probability into the running time formula for our algorithm
  using exact counting, we obtain the exact top-$k$ most frequent objects with
  probability at least~$1-\delta$ in the claimed running time.
\end{proof}

Note that the number of frequent objects decreases sharply with $s>1$, as
the~$k$-th most frequent one has a relative frequency of only
$\Ohsmall{k^{-s}}$.  The~$k^s H_{n,s}$ term in the communication volume is thus
small in practice, and, in fact, unavoidable in a sampling-based algorithm.
(One can easily verify this claim by observing that
this factor is the reciprocal of the $k$-th most frequent object's relative
frequency. It is clear that this object needs at least one occurrence in the
sample for the algorithm to be able to find it, and that the sample size must
thus scale with~$k^sH_{n,s}$.)

\subsection{Refinements}\label{ss:topkref}

When implementing such an algorithm, there are a number of
considerations to be taken into account to achieve optimal
performance. Perhaps most importantly, one should apply local
aggregation when inserting the sample into the distributed hash table
to reduce the amount of information that needs to be
communicated in practice. We now discuss other potential improvements.

\paragraph*{Choice of $\kopt$.} In practice, the choice of $\kopt$ in
Section~\ref{ss:topkcount} depends on the communication channel's
characteristics $\Tword$, and, to a lesser extent, $\Tstart$, in
addition to the problem parameters. Thus, an optimized implementation
should take them into account when determining the number of objects
to be counted exactly.

\paragraph*{Adaptive Two-Pass Sampling (Outline).} The objectives of the basic PAC
algorithm and its variant using exact counting could be unified as
follows: we sample in two passes. In the first pass, we use a small
sample size~$n\rho_0=f_0(n, k, \eps, \Delta)$ to determine the
nature of the input distribution. From the insights gained from this
first sample, we compute a larger sample size $n\rho_1 = f_1(n, k,
\hat{c_k}, \eps, \Delta)$. We then determine and return the~$k$
most frequent objects of this second sample.

Additionally, we can further refine this algorithm if we can already
tell from the first sample that with high probability, there is no
slope. If the absolute counts of the objects in the sample are large
enough to return the~$k$ most frequent objects in the sample with
confidence, then taking a second sample would be of little benefit and
we can return the~$k$ most frequent objects from the first sample.

\paragraph*{Using Distributed Bloom Filters.} Communication
efficiency of the algorithm using exact counting could be improved
further by counting sample elements with a distributed single-shot bloom filter
(dSBF)~\cite{SSM13} instead of a distributed hash table. We transmit
their hash values and locally aggregated counts. As multiple keys
might be assigned the same hash value, we need to determine the
element of rank $\kopt + \kappa$ instead of~$\kopt$, for some safety margin~$\kappa>0$. We request
the keys of all elements with higher rank, and replace the (hash,~value)
pairs with (key, value) pairs, splitting them where hash collisions
occurred. We now determine the element of rank $\kopt$ on the~%
$\kopt +\nobreak\kappa +\nobreak\#\text{collisions}$ elements. If an element whose
rank is at most $\kopt$ was part of the original $\kopt + \kappa$
elements, we are finished. Otherwise, we have to increase $\kappa$ to
determine the missing elements. Observe that if the frequent objects
are dominated by hash collisions, this implies that the input
distribution is flat and there exist a large number of nearly equally
frequent elements. Thus, we may not need to count additional elements
in this case.

\section[Top-k Sum Aggregation]{Top-$k$ Sum Aggregation}\label{s:sumagg}

\newcommand{\cm}{c_\text{max}}
\newcommand{\va}{v_\text{avg}}

Generalizing from Section~\ref{s:topk}, we now consider an input multiset of
keys with associated non-negative counts, and ask for the~$k$ keys whose counts
have the largest sums.  Again, the input~$M$ is distributed over the~$p$ PEs so
that no PE has more than~$\Oh{n/p}$ objects.  Define
$m \Def\nobreak \sum_{(k,v) \in M}{v}$ as the sum of all counts, and let no PE's
local sum exceed~$\Oh{m/p}$\footnote{This assumption is not strictly necessary,
  and the algorithm would still be communication efficient without it. However,
  making this assumption allows us to limit the number of samples transmitted by
  each PE as $s/p$ for global sample size~$s$.  To violate this criterion, a
  small number of PEs would have to hold~$\Om{s/p}$ elements that are likely to
  yield at least one sample, making up for a large part of the global sample
  size without contributing to the result.  In such a rare setting, we would
  incur up to a factor of~$p$ in communication volume and obtain running time
  $\Oh{\frac{n}{p} + \Tword \frac{\log p}{\eps}\sqrt{p \log\frac{n}{\delta}} +
    \Tstart \log{n}}.$}.
We additionally assume that the local input
and a key-aggregation thereof---e.g. a hash table mapping keys to their local
sums---fit into RAM at every PE.

It is easy to see that except for the sampling process, the algorithms of
Section~\ref{s:topk} carry over directly, but a different approach is required
in the analysis.

\subsection{Sampling}\label{s:sumagg:sampling}

Let~$s$ be the desired sample size, and define $\va \Def \sfrac{m}{s}$ as the
expected count required to yield a sample.  When sampling an object~$(k,v)$,
its expected sample count is thus~$\frac{v}{\va}$.  To retain constant time per
object, we add $\lfloor \frac{v}{\va}\rfloor$ samples directly, and one
additional sample with probability $\frac{v}{\va}-\lfloor\frac{v}{\va}\rfloor$
using a Bernoulli trial.  We can again use a geometric distribution to reduce
the number of calls to the random number generator.

To improve accuracy and speed up exact counting, we aggregate the local input in
a hash table, and sample the aggregate counts.  This allows us to analyze the
algorithms' error independent of the input objects' distribution.  A direct
consequence is that for each key and PE, the number of samples deviates from its
expected value by at most~$1$, and the overall deviation per
key~$\abs{s_i - \expect{s_i}}$ is at most~$p$.

\subsection{Probably Approximately Correct Algorithms}\label{s:sumagg:pac}

\begin{theorem}\label{thm:sumpac}
  We can compute an~$(\eps, \delta)$-approximation of the top-$k$ highest
  summing items in expected time
  $\Oh{\frac{n}{p} + \Tword \frac{\log{p}}{\eps}
    \sqrt{\frac{1}{p}\log\frac{n}{\delta}} + \Tstart \log{n}}$.
\end{theorem}

Sampling is done using local aggregation as described in
Section~\ref{s:sumagg:sampling}.  From then on, we proceed exactly as in
Algorithm \PAC from Section~\ref{ss:topkapprox}.

\begin{proof}
  The part of an element's sample count that is actually determined by sampling
  is the sum of up to~$p$ Bernoulli trials $X_1,\ldots,X_p$ with differing
  success probabilities.  Therefore, its expected value~$\mu$ is the sum of the
  probabilities, and we can use Hoeffding's inequality to bound the probability
  of significant deviations from this value.  Let~$X \Def \sum_{i=1}^{p}{X_i}$.
  Then,
  \begin{equation}
    \prob{\abs{X-\mu} \geq t} \leq 2e^{-\frac{2t^2}{p}}\punkt \label{eq:sa_hoeff}
  \end{equation}

  We now use this to bound the likelihood that an object has been very
  mis-sampled.  Consider an element~$j$ with exact sum~$x(j)$ and sample
  sum~$s_j$.  For some threshold~$\Delta$, consider the element mis-sampled
  if~$\abs{x(j)-s_j \va} \geq \frac{\Delta}{2}$, i.e. its estimated sum deviates
  from the true value by more than~$\frac{\Delta}{2}$ in either direction.
  Thus, we substitute~$t=\frac{\Delta}{2\va}$ into Equation~(\ref{eq:sa_hoeff})
  and bound the result by~$\frac{\delta}{n}$ to account for all elements.
  Solving for~$s=\frac{m}{\va}$, we
  obtain~$s \geq \frac{1}{\eps}\sqrt{2p\ln\frac{2n}{\delta}}$.

  In total, we require time \( \Oh{\frac{n}{p}} \)
  for sampling, \( \Ohsmall{\Tword \frac{s}{p} \log{p} + \Tstart \log{p}} \)
  for insertion (as no PE's local sum exceeds \(\Ohsmall{m/p}\),
  none can yield more than \(\Ohsmall{s/p}\)
  samples), and \( \Ohsmall{\Tword \frac{s}{p} + \Tstart \log{n}} \)
  for selection.  We then obtain the claimed bound as sum of the components.
\end{proof}

As in Section~\ref{s:topk}, we can use exact summation to obtain a more precise
answer.  We do not go into details here, as the procedure is nearly identical.
The main difference is that a lookup in the local aggregation result now
suffices to obtain exact local sums without requiring consultation of the input.

\section{Data Redistribution}\label{s:redistribution}

Let $n_i$ denote the number of data objects present at~PE $i$. Let~$n=\sum_i
n_i$.  We want to redistribute the data such that afterwards each~PE has at most
$\navg=\ceil{n/p}$ objects and such that~PEs with more than~$\navg$ objects only send
data (at most~$n_i-\navg$ objects) and~PEs with at most~$\navg$ objects only receive
data (at most~$\navg-n_i$ objects).  We split the~PEs into separately numbered
groups of senders~$s_i$ and receivers~$d_i$.  We also compute the deficit~%
$d_i\Is\navg-n_i$ on receivers and the surplus~$s_i\Is n_i-\navg$ on
senders. Then we compute the prefix sums~$d$ and~$s$ of these sequences (i.e.,~%
$d_i\Is\sum_{j\leq i}d_j$ and~$s_i\Is\sum_{j\leq i}s_j$).  Effectively,~$d$
enumerates the empty slots able to receive objects and~$s$ enumerates the
elements to be moved. Now we match receiving slots and elements to be moved by
merging the sequences~$d$ and~$s$.  This is possible in time~$\Oh{\alpha\log p}$
using Batcher's parallel merging algorithm~\cite{Batcher68}. A subsequence of
the form~$\seq{d_i,s_j,\ldots,s_{j+k},d_{i+a},\ldots,d_{i+b}}$ indicates that
sending~PEs~$j,\ldots,j+k$ move their surplus to receiving~PE~$i$ (where sending
PE~$j+k$ only moves its items numbered~$s_{j+k-1}+1..d_{i+a}-1$). This is a
gather operation.  Sending~PE~$j+k$ moves its remaining elements to receiving
PEs~$i+a..i+b$. This is a scatter operation. These segments of~PE numbers can be
determined using segmented prefix operations~\cite{Ble89}. Overall, this can be
implemented to run in time~$\Oh{\Tword \max_in_i+\Tstart\log p}$. Even though this operation
cannot remove worst case bottlenecks, it can significantly reduce network
traffic.

\section{Experiments}\label{s:exp}

\renewcommand{\figurename}{Figure}
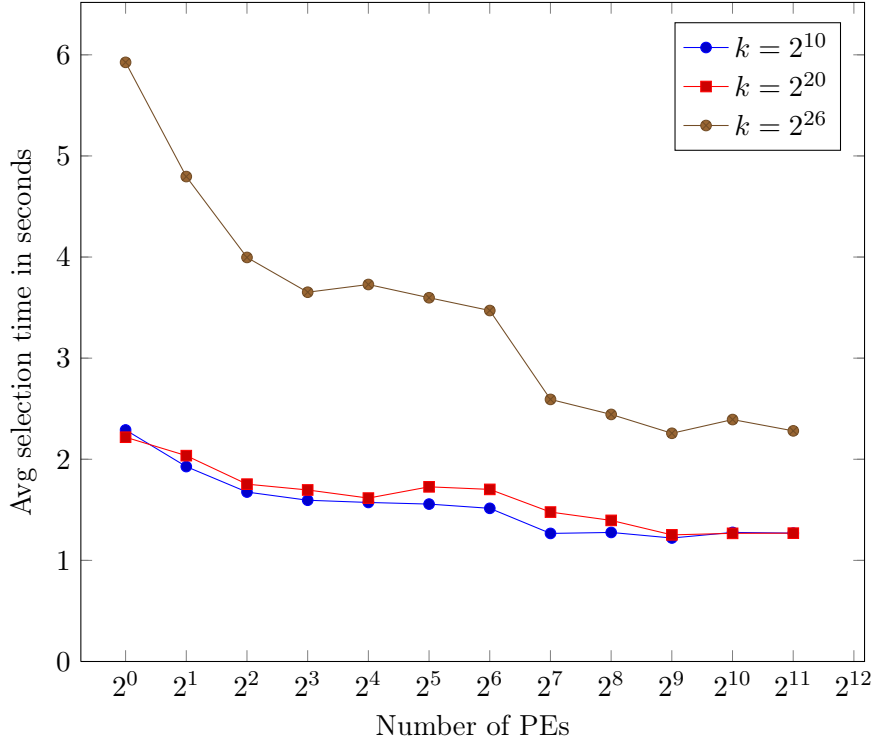
\begin{figure}
\centering
\begin{tikzpicture}
  \begin{semilogxaxis}[
    width=12cm,
    log basis x=2,
    xmin=0.6,
    ymin=0,
    minor ytick={0,1,2,3,4,5,6},
    legend pos=north east,
    xlabel={Number of PEs},
    ylabel={Avg selection time in seconds}
  ]
  \addplot coordinates { (1,2.289902) (2,1.92767) (4,1.67565) (8,1.59483) (16,1.57288) (32,1.55677) (64,1.51482) (128,1.26662) (256,1.27629) (512,1.22115) (1024,1.2766) (2048,1.269321) };
  \addlegendentry{$k=2^{10}$};
  \addplot coordinates { (1,2.21813) (2,2.03585) (4,1.75388) (8,1.69641) (16,1.616085) (32,1.72745) (64,1.7021) (128,1.47777) (256,1.39552) (512,1.25228) (1024,1.26745) (2048,1.2692) };
  \addlegendentry{$k=2^{20}$};
  \addplot coordinates { (1,5.92552) (2,4.795934) (4,3.99602) (8,3.65208) (16,3.72848) (32,3.59751) (64,3.47115) (128,2.59199) (256,2.44343) (512,2.25697) (1024,2.39273) (2048,2.28088) };
  \addlegendentry{$k=2^{26}$};
 \end{semilogxaxis}
\end{tikzpicture}
\caption{Weak scaling plot for selecting the $k$-th largest object, $n/p=2^{28}$, Zipf distribution.}
\label{fig:sel_1}
\end{figure}

We now present an experimental evaluation of the unsorted selection from
Section~\ref{ss:unsortedSelection} and the top-$k$ most frequent objects
algorithms from Section~\ref{s:topk}.

\paragraph{Experimental Setup.} All algorithms were implemented in C++11 using
OpenMPI~1.8 and Boost.MPI~1.59 for inter-process communication.  Additionally,
Intel's Math Kernel Library in version~11.2 was used for random number
generation.  All code was compiled with the \texttt{clang++} compiler in
version~3.7 using optimization level \texttt{-Ofast} and instruction set
specification \texttt{-march=sandybridge}.  The experiments were conducted on
InstitutsCluster II at Karlsruhe Institute of Technology, a distributed system
consisting of 480 computation nodes, of which 128 were available to us.  Each
node is equipped with two Intel Xeon E5-2670 processors for a total of 16 cores
with a nominal frequency of 2.6\,GHz, and 64\,GiB of main memory.  In total,
2048 cores were available to us. An Infiniband 4X QDR interconnect provides
networking between the nodes.

\paragraph{Methodology.} We run \emph{weak scaling} benchmarks, which show how
wall-time develops for fixed per-PE problem size~$n/p$ as~$p$ increases.  We
consider~$p=1$ to $2048$~PEs, doubling~$p$ in each step.  Each PE is mapped to
one physical core in the cluster.

\paragraph{Zipf's Law}states that the frequency of an object from a
multiset~$M$ with~$|M|=n$ is inversely proportional to its rank among the
objects ordered by frequency.  Here, we consider the general case with exponent
parameter~$s$, i.e.~$x_i= n\frac{i^{-s}}{H_{n,s}}$,
where~$H_{n,s}\!=\sum_{i=1}^{n}{i^{-s}}$ is the~$n$-th \textit{generalized
  harmonic number}.

\subsection{Unsorted Selection}\label{s:exp:sel}

\paragraph{Input Generation.} We select values from the high tail of Zipf
distributions.  Per PE, we consider~$2^{24}$, $2^{26}$, and~$2^{28}$ integer
elements.  To test with non-uniformly distributed data, the PE's distribution
parameters are randomized.  The Zipf distributions comprise
between~$2^{20}-2^{16}$ and~$2^{20}$ elements, with each PE's value chosen
uniformly at random. Similarly, the exponent~$s$ is uniformly distributed
between~$1$ and~$1.2$. This ensures that several PEs contribute to the result,
so that the distribution is asymmetric, without the computation becoming a local
operation at one PE, which has all of the largest elements.

We used several values of~$k$, namely $1024$, $2^{20}$, and $2^{26}$.  We do not
consider smaller values than 1024, as for values this small, it would be more
efficient to locally select the~$k$ largest (or smallest) elements, and run a
simple distributed selection on those.

\paragraph{Results.} Figure~\ref{fig:sel_1} shows the results for selecting the
$k$-th largest values from the input, for the above values of~$k$.  We can see
that in most cases, the algorithm scales even better than the bounds lead us to
expect---running time decreases as more PEs are added.  This is especially
prominent when selecting an element of high rank ($k=2^{26}$ in
Figure~\ref{fig:sel_1}).  The majority of the time is spent in partitioning,
i.e. local work, dominating the time spent on communication.  This underlines
the effect of communication efficiency.

\subsection[Top-k Most Frequent Objects]{Top-$k$ Most Frequent  Objects}
\label{s:exp:mfo}

As we could not find any competitors to compare our methods against, we use two
naive centralized algorithm as baseline.  The first algorithm, \Naive, samples
the input with the same probability as algorithm PAC, but instead of using a
distributed hash table and distributed selection, each PE sends its aggregated
local sample to a coordinator.  The coordinator then uses quickselect to
determine the elements of rank~$1..k$ in the global sample, which it returns.  Algorithm
\NaiveTree proceeds similarly, but uses a tree reduction to send the elements to
the coordinator to reduce latency.  Similar to Algorithm \PAC's hash table
insertion operation, this reduction aggregates the counts in each step to keep
communication volume low.

\paragraph{Input Generation.} We consider~$2^{24}$,~$2^{26}$ and~$2^{28}$
elements per PE, which are generated according to different random
distributions.  First, we consider elements distributed according to Zipf's Law
with~$2^{20}$ possible values.  These values are very concentrated and model
word frequencies in natural languages, city population sizes, and many other
rankings well~\cite{auerbach1913,zipf1935}, the most frequent element
being~$j$-times more frequent than that of rank~$j$.  Next, we study a negative
binomial distribution with~$r=1000$ and success probability~$p=0.05$.  This
distribution has a rather wide plateau, resulting in the most frequent objects
and their surrounding elements all being of very similar frequency.  For
simplicity, each PE generates objects according to the same distribution, as the
distributed hash table into which the sample is inserted distributes elements
randomly.  Thus, tests with non-uniformly distributed data would not add
substantially to the evaluation.

\paragraph{Approximation Quality.} To evaluate different accuracy levels, we
consider the~$(\eps,\delta)$ pairs $(3\cdot10^{-4},10^{-4})$ and
$(10^{-6},10^{-8})$.  This allows us to evaluate how running time develops under
different accuracy requirements.

We then select the~$k=32$ most frequent elements from the input according to the
above requirements.  We do not vary the parameter~$k$ here, as it has very
little impact on overall performance.  Instead, we refer to
Section~\ref{s:exp:sel} for experiments on unsorted selection, which is the only
subroutine affected by increasing~$k$ and shows no increase up to~$k=2^{20}$.

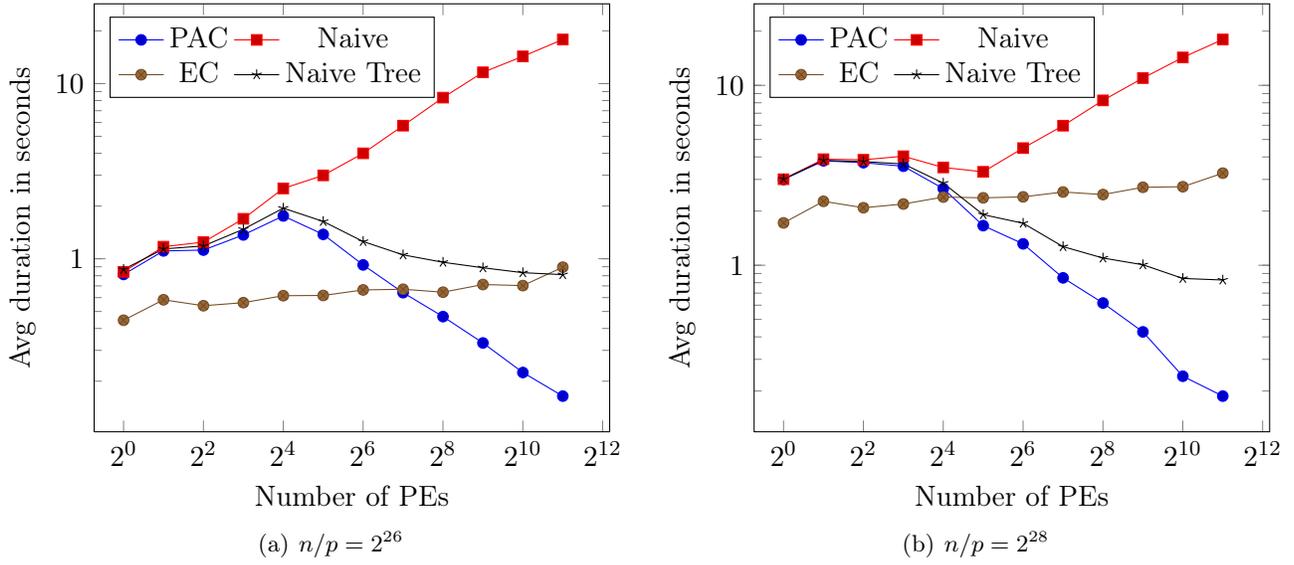
\begin{figure}
\centering
\mbox{
\subfigure[$n/p=2^{26}$]{
  \label{fig:ws_1a}
  \hspace{-15pt}
  \begin{tikzpicture}
    \begin{axis}[
      xmode=log,
      ymode=log,
      log basis x=2,
      xmin=0.6,
      yticklabel=\pgfmathparse{exp(\tick)}\pgfmathprintnumber{\pgfmathresult},
      legend columns=2,
      legend pos=north west,
      xlabel={Number of PEs},
      ylabel={Avg duration in seconds}
    ]
    \addplot coordinates { (1,0.813626) (2,1.10855) (4,1.12079) (8,1.36782) (16,1.75582) (32,1.37717) (64,0.922133) (128,0.640385) (256,0.467409) (512,0.33065) (1024,0.224177) (2048,0.164381) };
    \addlegendentry{PAC};
    \addplot coordinates { (1,0.841938) (2,1.17159) (4,1.2442) (8,1.69068) (16,2.51978) (32,2.99248) (64,3.99828) (128,5.75208) (256,8.31361) (512,11.595) (1024,14.2989) (2048,17.8506) };
    \addlegendentry{Naive};
    \addplot coordinates { (1,0.445627) (2,0.58317) (4,0.539362) (8,0.561527) (16,0.615429) (32,0.617268) (64,0.663897) (128,0.669805) (256,0.644336) (512,0.712624) (1024,0.702662) (2048,0.894973) };
    \addlegendentry{EC};
    \addplot coordinates { (1,0.868986) (2,1.13795) (4,1.18336) (8,1.47349) (16,1.9443) (32,1.63324) (64,1.25456) (128,1.05254) (256,0.954903) (512,0.888346) (1024,0.832603) (2048,0.811554) };
    \addlegendentry{Naive Tree};
   \end{axis}
  \end{tikzpicture}
}

\subfigure[$n/p=2^{28}$]{
  \label{fig:ws_1b}
  \begin{tikzpicture}
  \begin{axis}[
    xmode=log,
    ymode=log,
    log basis x=2,
    xmin=0.6,
    yticklabel=\pgfmathparse{exp(\tick)}\pgfmathprintnumber{\pgfmathresult},
    legend columns=2,
    legend pos=north west,
    xlabel={Number of PEs},
    ylabel={Avg duration in seconds}
  ]
    \addplot coordinates { (1,2.99524) (2,3.80641) (4,3.71681) (8,3.55206) (16,2.66817) (32,1.66075) (64,1.31528) (128,0.851002) (256,0.616839) (512,0.426267) (1024,0.241647) (2048,0.187379) };
    \addlegendentry{PAC};
    \addplot coordinates { (1,3.00774) (2,3.88709) (4,3.85763) (8,4.03611) (16,3.49701) (32,3.30599) (64,4.47939) (128,5.96703) (256,8.24847) (512,10.973) (1024,14.2847) (2048,17.9784) };
    \addlegendentry{Naive};
    \addplot coordinates { (1,1.71952) (2,2.2667) (4,2.08766) (8,2.18973) (16,2.39353) (32,2.37044) (64,2.39986) (128,2.55326) (256,2.47119) (512,2.7129) (1024,2.73194) (2048,3.25047) };
    \addlegendentry{EC};
    \addplot coordinates { (1,3.01669) (2,3.82793) (4,3.76896) (8,3.65824) (16,2.86093) (32,1.91097) (64,1.71196) (128,1.26903) (256,1.0968) (512,1.0062) (1024,0.844812) (2048,0.827931) };
    \addlegendentry{Naive Tree};
  \end{axis}
  \end{tikzpicture}
}
}
\caption{Weak scaling plot for computing the 32 most frequent objects,
  $\eps=3\cdot10^{-4}$, $\delta=10^{-4}$. EC suffers from constant overhead for
  exact counting.}
\label{fig:ws_1}
\end{figure}

\begin{figure}
\centering
\begin{tikzpicture}
  \begin{loglogaxis}[
    width=11cm,
    log basis x=2,
    xlabel={Number of PEs},
    ylabel={Avg duration in seconds},
    xmin=0.6,
    ytick={1,3,6,10,30,60,100},
    minor ytick={2,3,4,5,6,7,8,9,10,20,30,40,50,60,70,80,90,100,110,120,130,140,150},
    yticklabel=\pgfmathparse{exp(\tick)}\pgfmathprintnumber{\pgfmathresult},
    legend columns=2,
    legend pos=north west,
  ]
  \addplot coordinates { (1,2.91721) (2,3.78894) (4,3.73834) (8,4.67459) (16,6.06402) (32,6.11165) (64,6.11105) (128,6.18752) (256,6.22839) (512,6.31059) (1024,6.36118) (2048,6.44885) };
  \addlegendentry{PAC};
  \addplot coordinates { (1,2.93317) (2,3.92576) (4,3.83584) (8,5.02242) (16,6.8034) (32,7.82993) (64,9.5926) (128,13.287) (256,20.2857) (512,35.3482) (1024,63.2194) (2048,120.177) };
  \addlegendentry{Naive};
  \addplot coordinates { (1,2.07282) (2,2.71298) (4,2.52202) (8,3.13276) (16,4.12155) (32,4.08709) (64,4.07857) (128,4.07005) (256,4.07844) (512,4.09159) (1024,4.13643) (2048,4.21051) };
  \addlegendentry{EC};
  \addplot coordinates { (1,2.96499) (2,3.86641) (4,3.73763) (8,4.76514) (16,6.15482) (32,6.33491) (64,6.36312) (128,6.4643) (256,6.55763) (512,6.63166) (1024,6.69568) (2048,6.76588) };
  \addlegendentry{Naive Tree};
 \end{loglogaxis}
\end{tikzpicture}
\caption{Weak scaling plot for computing the 32 most frequent objects,
  $n/p=2^{28}$, $\eps=10^{-6}$, $\delta=10^{-8}$. Only EC can use sampling,
  other methods must consider all objects.}
\label{fig:ws_2}
\end{figure}
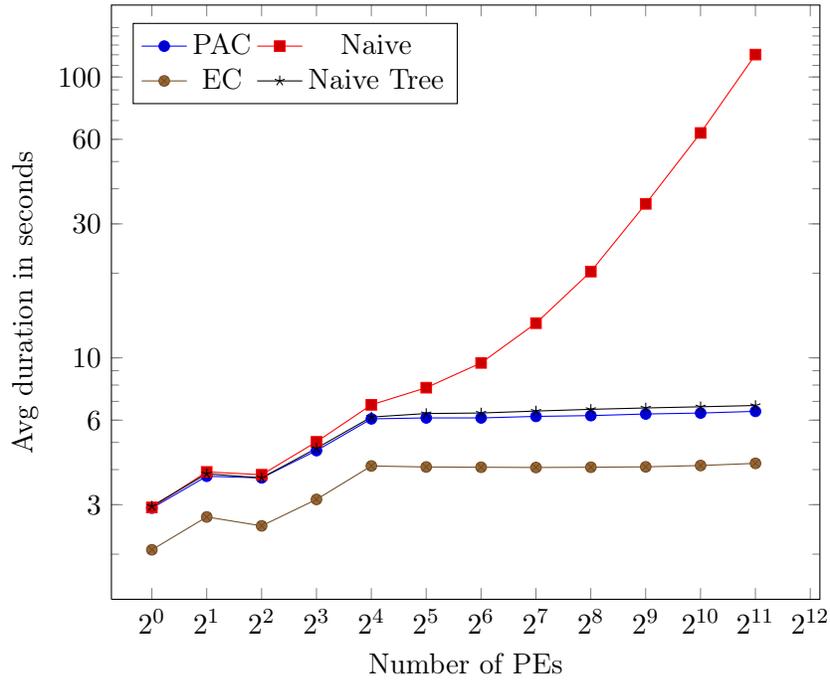

\paragraph{Results.} Figure~\ref{fig:ws_1} shows the results for~$2^{28}$
elements per PE using~$\eps=3\cdot 10^{-4}$ and~$\delta=10^{-4}$.  We can
clearly see that Algorithm \Naive does not scale beyond a single node at all
($p > 16$).  In fact, its running time is directly proportional to~$p$, which is
consistent with the coordinator receiving~$p-1$ messages---every other PE sends
its key-aggregated sample to the coordinator.  Algorithm \NaiveTree fares
better, and actually improves as more PEs are added.  This is easily explained
by the reduced sample size per PE as~$p$ increases, decreasing sampling time.
However, communication time begins to dominate, as the decrease in overall
running time is nowhere near as strong as the decrease in local sample size.
This becomes clear when comparing it to Algorithm \PAC, which outperforms
\NaiveTree for any number of PEs.  We can see that it scales nearly
perfectly---doubling the number of PEs (and thereby total input size)  roughly
halves running time.  Since these three algorithms all use the same sampling
rate, any differences in running time are completely due to time spent on
communication.

Lastly, let us consider Algorithm \EC.  In the beginning, it benefits from its
much smaller sample size (see Section~\ref{ss:topkcount}), but as~$p$ grows, the
local work for exact counting dominates overall running time strongly and
Algorithm EC is no longer competitive.  Since local work remains constant with
increasing~$p$, we see nearly no change in overall running time.  To see the
benefits of Algorithm \EC, we need to consider stricter accuracy requirements.

In Figure~\ref{fig:ws_2}, we consider~$\eps=10^{-6}$ and~$\delta=10^{-8}$.  For
Algorithms \PAC, \Naive, and \NaiveTree, this requires considering the entire
input for any number of PEs, as sample size is proportional
to~$\frac{1}{\eps^2}$, which the other terms cannot offset here.  Conversely,
Algorithm \EC's sample size depends only linearly on~$\eps$, resulting in sample
sizes orders of magnitude below those of the other algorithms.

Again, we can see that Algorithm \Naive is completely unscalable.  Algorithm
\NaiveTree performs much better, with running times remaining roughly constant
at around 6.5 seconds as soon as multiple nodes are used.  Algorithm \PAC
suffers a similar fate, however it is slightly faster at 6.2 seconds.  This
difference stems from reduced communication volume.  However, both are dominated
by the time spent on aggregating the input.  Lastly, Algorithm \EC is
consistently fastest, requiring 4.1 seconds, of which 3.7 seconds are spent on
exact counting\footnote{When not all cores are used, memory bandwidth per core
  is higher.  This allows faster exact counting for $p=1$ to 8 cores on a single
  node.}.  This clearly demonstrates that Algorithm \EC is superior for
small~$\eps$.

Smaller local input sizes do not yield significant differences, and preliminary
experiments with elements distributed according to a negative binomial
distribution proved unspectacular and of little informational value, as the
aggregated samples have much fewer elements than in a Zipfian distribution---an
easy case for selection.

\section{Conclusions}\label{s:conclusions}

We have demonstrated that a variety of top-$k$ selection problems can be solved
in a communication efficient way, with respect to both communication volume and
latencies.  The basic methods are simple and versatile---the owner-computes
rule, collective communication, and sampling.  Considering the significant
previous work on some of these problems, it is a bit surprising that such simple
algorithms give improved results for such fundamental problems.  However, it
seems that the combination of communication efficiency and parallel scalability
has been neglected for many problems.  Our methods might have particular impact
on applications where previous work has concentrated on methods with a pure
master--worker scheme.

It is therefore likely that our approach can also be applied to further
important problems.  For example, distributed streaming algorithms that
generalize the centralized model of Yi and Zhang~\cite{YiZhang13} seem very
promising.  The same holds for lower bounds, which so far have also neglected
multiparty communication with point-to-point communication (see also
\cite{SSM13}).

Closer to the problems considered here, there is also a number of interesting
open questions.  For the sorted selection problem from Section~\ref{ss:musel}, it
would be interesting to see whether there is a scalable parallel algorithm that
makes an information theoretically optimal number of comparisons as in the
sequential algorithm of Varman et al.~\cite{VSIR91}.  Our analysis of
approximate multiselection ignores the case where~$\kmax-\nobreak\kmin=o(\kmax)$.  It
can probably be shown to run in expected
time~$\Ohsmall{\Tstart\log p\log\frac{\kmax}{\kmax-\kmin}}$.
For the multicriteria top-$k$ problem from Section~\ref{s:multic}, we could
consider parallelization of advanced algorithms that scan less elements and perform
less random accesses, such as \cite{BMSTW06}.

Regarding the top-$k$ most frequent objects and sum aggregation, we expect to be
able to conduct fully distributed monitoring queries without a substantial
increase in communication volume over our one-shot algorithm.




\bibliographystyle{plain}
\bibliography{diss.bib}
\newpage
\begin{appendix}

\section[Multisequence Selection]{Multisequence Selection \normalsize\rm\cite{ABSS15}}\label{ss:msselect}
\renewcommand{\figurename}{Algorithm}
\begin{figure}
\begin{code}
\Rem{select object with global rank $k$}\\
\Procedure msSelect$(s, k)$\+\\
  \If $\sum_{1\leq i\leq p}|s@i| = 1$ \Then \RRem{base case}\+\\
    \Return the only nonempty object\-\\
  select a pivot $v$ uniformly at random\\
  find $j$ such that $s[1..j] < v$ and $s[j+1..]\geq v$\-\\
  \If $\sum_{1\leq i\leq p}|j@i| \geq k$ \Then\+\\ \Return msSelect$(s[1..j], k)$\-\\
  \Else\+\\ \Return msSelect$(s[j_1+1..], k-\sum_{0\leq i<p}|j@i|)$\\
\end{code}
\caption{\label{alg:musel}Multisequence selection.}
\end{figure}
Figure~\ref{alg:musel} gives high level pseudo code.  The base case occurs if
there is only a single object (and~$k=1$).  We can also restrict the search to
the first $k$ objects of each local sequence. A random object
is selected as a pivot. This can be done in parallel by choosing the same random
number between~1 and~$\sum_{i}|s@i|$ on all~PEs. Using a prefix sum over the
sizes of the sequences, this object can be located easily in time~%
$\Oh{\Tstart\log p}$. Where ordinary quickselect has to partition the input
doing linear work, we can exploit the sortedness of the sequences to obtain the
same information in time $\Oh{\log \sigma}$ with $\sigma := \max_i|s@i|$ by doing binary
search in parallel on each~PE.\@ If items are evenly distributed, we have~$\sigma =
\Thsmall{\min(k,\frac{n}{p})}$, and thus only time
$\Ohsmall{\log\min(k,\frac{n}{p})}$ for the search, which partitions all the
sequences into two parts.  Deciding whether we have to continue searching in the
left or the right parts needs a global reduction operations taking time
$\Oh{\Tstart\log p}$. As in ordinary quickselect, the expected depth of the recursion is
$\Ohsmall{\log\sum_i|d_i|}=\Oh{\log\min(kp,n)}$. We obtain the following result.

\begin{theorem}\label{thm:msselect}
Algorithm~\ref{alg:musel} can be implemented to run in expected time
$$\Oh{\left( \Tstart\log{p} + \log{\min{\left(\frac{n}{p}, k\right)}} \right)
  \cdot \log{\min{\left(kp, n\right)}}} = \Oh{\Tstart\log^2kp}\punkt$$
\end{theorem}

\section{Running Times of Existing Algorithms}\label{app:prevproof}

We now prove the running times given for previous works in
Table~\ref{tbl:results}.


\subsubsection{Unsorted Selection}

Previous algorithms, see \cite{San98a-e}, rely on randomly distributed input
data or they explicitly redistribute the data.  Thus, they have to move
$\Omsmall{\sfrac{n}{p}}$ elements in the worst case.  The remainder of the bound
follows trivially.  \qed


\subsection{Sorted Selection}

We could not find any prior results on distributed multiselection from
sorted lists and thus list a sequential result by Varman et al.~\cite{VSIR91}.


\subsection{Bulk Parallel Priority Queue}

The result in \cite{San98a-e} relies on randomly distributed input
data. Therefore, in operation~\textsf{insert*}, each PE needs to send its
$\Ohsmall{k/p}$ elements to random PEs, costing $\Ohsmall{(\Tstart +
  \Tword)\frac{k}{p}}$ time. Then, operation \textsf{deleteMin*} is fairly
straight-forward and mostly amounts to a selection.  The deterministic
parallel heap \cite{DeoPra92} needs to sort inserted elements and then they
travel through $\log\sfrac{n}{p}$ levels of the data structure, which is
allocated to different PEs. This alone means communication cost
$\Om{\Tword\sfrac{k}{p}\log p}$.


\subsection{Heavy Hitters Monitoring}

Huang et al.\ give a randomized heavy hitters monitoring
algorithm~\cite{HuaYiZha2012} that, \emph{for constant failure
  probability~$\eps$}, requires time
$\Oh{\frac{n}{p} + \Tstart \frac{\sqrt{p}}{\gamma}\log{n}}$ in our model.

\medskip\noindent
\begin{proof}
  All communication is between the controller node and a monitor node, thus the
  maximum amount of communication is at the controller node.  Each update, which
  consists of a constant number of words, is transmitted separately.  Thus, the
  communication term given by the authors transfers directly into our model
  (except that the number of monitor nodes~$k$ is~$p-1$ here).
\end{proof}


\subsection{Top-k Frequent Objects Monitoring}

Monitoring Query 1 of~\cite{BabOls03short} performs top-$k$ most frequent object
monitoring, for which it incurs a running time of
$\Om{\frac{n}{p} + \Tword \frac{k}{\gamma} + \Tstart\frac{1}{\gamma}}$ for
relative error bound~$\gamma=\frac{n}{\Delta}$, where $\Delta$ corresponds
to~$\eps$ in their notation.  This algorithm also has further restrictions: It
does not provide approximate counts of the objects in the top-$k$ set.  It can
only handle a small number~$N$ of distinct objects, all of which must be known
in advance.  It requires~$\Th{Np}$ memory on the coordinator node, which is
prohibitive if $N$ and $p$ are large.  It must also be initialized with a
top-$k$ set for the beginning of the streams, using an algorithm such as
TA~\cite{FLN03}.  We now present a family of inputs for which the algorithm uses
the claimed amount of communication.

Initialize with~$N=k+2$ items~$O_1,\ldots,O_{k+2}$ that all have the same
frequency.  Thus, the initial top-k set comprises an arbitrary~$k$ of these.
Choose one of the two objects that are \emph{not} in the top-$k$ set, and refer
to this object as~$O_s$, and pick a peer (PE)~$N_f$.  Now, we send~$O_s$
to~$N_f$ repeatedly, and all other items to all other peers in an evenly
distributed manner (each peer receives around the same number of occurrences of
each object).  After at most $2\Delta$ steps, the top-$k$ set has become invalid
and an expensive full resolution step is required\footnote{This actually happens
  much earlier, the first ``expensive'' resolution is required after
  only~$\Ohsmall{\frac {\Delta}{p}}$ steps.  This number increases, but will
  never exceed $\sum_{i=0}^{\infty}{\Delta (k+1)^{-i}} \leq 2\Delta$.}.  As we
expect~$\Delta$ to be on the large side, we choose~$F_0=0$
and~$F_j=\frac{1}{p-1}$ for~$j>0$ as per the instructions
in~\cite{BabOls03short}.  We can repeat this cycle to obtain an instance of the
example family for any input size~$n$.  Note that the number of ``cheap''
resolution steps during this cycle depends on the choice of the~$F_j$ values,
for which Babcock and Olston give rough guidelines of what might constitute a
good choice, but do not present a theoretical analysis of how it influences
communication.  Here, we ignore their cost and focus solely on the cost of
``expensive'' resolutions.

By the above, an ``expensive'' resolution round is required every (at
most)~$\Oh{p\Delta}$ items in the input.  Since the resolution set
contains~$k+1$ objects (the top-$k$ set plus the non-top-$k$ object with a
constraint violation), each ``expensive'' resolution has communication
cost~$\Th{\Tword k p + \Tstart p}$.  Thus, we obtain a total communication cost
for expensive resolutions of
$\Omsmall{\frac{1} {p\gamma} (\Tword kp + \Tstart
  p)}=\Omsmall{\Tword\frac{k}{\gamma} + \Tstart\frac{1}{\gamma}}$,
for a given relative error bound~$\gamma$.  Additionally, each item requires at
least constant time in local processing, accounting for the additive
$\Omsmall{\frac{n}{p}}$ term.  The actual worst-case communication cost is
likely even higher, but this example suffices to show that the approach
of~\cite{BabOls03short} is not communication-efficient in our model.  \qed


\subsection{Multicriteria Top-k}

Previous algorithms such as TPUT~\cite{CaoWang04} or
KLEE~\cite{KLEE2005} are not directly comparable to our approach for a
number of reasons. First, they do not support arbitrary numbers of
processing elements, but limit $p$ to the number of criteria $m$. Each~%
PE is assigned one or more complete index lists, whereas our approach
splits the \emph{objects} among the~PEs, storing the index lists of
the locally present objects on each~PE. This severely limits TPUT's
and KLEE's scalability. Secondly, as noted in the introduction, these
algorithms use a centralized master--worker approach, where the master
(or \emph{coordinator}) node handles \emph{all} communication. This
further limits scalability and leads to an inherent increase in
communication volume by a factor of up to~$p$. Thirdly, they are
explicitly designed for \emph{wide-area networks} (WANs), whereas our
algorithms are designed with strongly interconnected PEs in mind, as
they might be found in a data center. Since the modeling assumptions
are too different to provide a meaningful comparison, we refrain from
giving a communication analysis of these algorithms \mbox{(nor was one
provided in the original papers).}

\end{appendix}

\end{document}